\documentclass{amsart}
\usepackage{amsfonts}
\usepackage{amsmath}
\usepackage{amssymb}

\theoremstyle{plain}
\newtheorem{theorem}{Theorem}[section]
\newtheorem{lemma}[theorem]{Lemma}
\newtheorem{proposition}[theorem]{Proposition}
\theoremstyle{definition}
\newtheorem{definition}[theorem]{Definition}

\theoremstyle{remark}
\newtheorem{remark}[theorem]{Remark}
\numberwithin{equation}{section}

\begin{document}
\title{On the complete integrability of nonlinear dynamical systems on
discrete manifolds within the gradient-holonomic approach}
\author{Yarema A. Prykarpatsky}
\address{ The Ivan Franko State Pedagogical University, Drohobych, Lviv
region, Ukraine, and the Department of Applied Mathematics at the Agrarian
University of Krakow, Poland}
\email{yarpry@gmail.com}
\author{Nikolai N. Bogolubov (Jr.)}
\address{The V.A. Steklov Mathematical Institute of RAS, Moscow, Russian
Federation and the Abdus Salam International Centre of Theoretical Physics,
Trieste, Italy}
\email{nikolai\_bogolubov@hotmail.com}
\author{Anatoliy K. Prykarpatsky}
\address{The Department of Mining Geodesy at the AGH University of Science
and Technology, Crac\'{o}w 30059, Poland, and the Ivan Franko State
Pedagogical University, Drohobych, Lviv region, Ukraine}
\email{pryk.anat@ua.fm, prykanat@agh.edu.pl}
\author{Valeriy H. Samoylenko}
\address{The Department of Mechanics and Mathematics at the T. Shevchenko
National University, Kyiv, Ukraine}
\email{vsam@univ.kiev.ua}
\maketitle

\begin{abstract}
A gradient-holonomic approach for the Lax type integrability analysis of
differential-discrete dynamical systems is described. The asymptotical
solutions to the related Lax equation are studied, the related gradient
identity subject to its relationship to a suitable Lax type spectral problem
is analyzed in detail. The integrability of the discrete nonlinear Schr\"{o}%
dinger, Ragnisco-Tu and Burgers-Riemann type dynamical systems is treated,
in particular, their conservation laws, compatible Poissonian structures and
discrete Lax type spectral problems are obtained within the
gradient-holonomic approach.
\end{abstract}

\section{Preliminary notions and definitions}

Consider an infinite dimensional discrete manifold $M\subset l_{2}(%
\mathbb{Z}
;%
\mathbb{C}
^{m})$ for some integer $m\in
\mathbb{Z}
_{+}$ and a general nonlinear dynamical system on it in the form
\begin{equation}
dw/dt=K[w],  \label{D1}
\end{equation}%
where $w\in M$ \ and $K$ $:$ $M\rightarrow T(M)$ is a Frechet smooth
nonlinear local functional on $M$ and $t\in \mathbb{R}$ is the evolution
parameter. As examples of dynamical systems (\ref{D1}) on a discrete
manifold $M\subset l_{2}(%
\mathbb{Z}
;%
\mathbb{C}
^{2})$ one can consider the well-known \cite{AL,MS} discrete nonlinear Schr%
\"{o}dinger equation%
\begin{equation}
\left.
\begin{array}{c}
du_{n}/dt=i(u_{n+1}-2u_{n}+u_{n-1})-iv_{n}u_{n}(u_{n+1}+u_{n-1}) \\
dv_{n}/dt=-i(v_{n+1}-2v_{n}+v_{n-1})+iv_{n}u_{n}(v_{n+1}+v_{n-1})%
\end{array}%
\right\} :=K_{n}[u,v],  \label{D2}
\end{equation}%
the so called Ragnisco-Tu \cite{RT} equation
\begin{equation}
\left.
\begin{array}{c}
du_{n}/dt=u_{n+1}-u_{n}^{2}v_{n} \\
dv_{v}/dt=-v_{n-1}+u_{n}v_{n}^{2}%
\end{array}%
\right\} :=K_{n}[u,v],  \label{D2a}
\end{equation}%
where we put $w:=(u,v)^{\intercal }\in M,$ \ and \ the inviscid
Riemann-Burgers equation \cite{Kup1} on a discrete manifold $M\subset l_{2}(%
\mathbb{Z}
;\mathbb{R}):$
\begin{equation}
dw_{n}/dt=w_{n}(w_{n+1}-w_{n-1})/2:=K_{n}[w]  \label{D2b}
\end{equation}%
and its Riemann type \cite{GBPPP,PAPP} generalizations, where $w\in M,$
having applications \cite{DEGM} in diverse physics investigations.

For studying the integrability properties of differential-difference
dynamical system \ (\ref{D1}) we will develop below a gradient-holonomic
scheme before devised in \cite{PM,MBPS,BPS,HPP} for nonlinear dynamical
systems defined on spatially one-dimensional functional manifolds and
extended in \cite{Pr} on the case of discrete manifolds.

Denote by $(\cdot ,\cdot )$ \ the standard bi-linear form on the space $%
T^{\ast }(M)\times T(M)$ naturally induced by that existing in the Hilbert
space $l_{2}(%
\mathbb{Z}
;%
\mathbb{C}
^{m}).$ Having denoted by $\mathcal{D}(M)\mathcal{\ }$smooth functionals on $%
M,$ for any functional $\gamma \in \mathcal{D}(M)$ one can define the
gradient $\mathrm{grad}$ $\gamma \lbrack w]\in T^{\ast }(M)$ as follows:%
\begin{equation}
\mathrm{grad}\text{ }\gamma \lbrack w]:=\gamma ^{\prime ,\ast }[w]\cdot 1,
\label{D3}
\end{equation}%
where the dash-sign $"^{\prime }"$ means the corresponding Frechet
derivative and the star-sign $"\ast "$ means the conjugation naturally
related with the bracket on $T^{\ast }(M)\times T(M).$

\begin{definition}
\label{Df_1} A linear smooth operator $\vartheta :T^{\ast }(M)\times T(M)$
is called \textit{Poissonian} on the manifold $M,$ if the bi-linear bracket
\begin{equation}
\{\cdot ,\cdot \}_{_{\vartheta }}:=(\mathrm{grad}\text{ }(\cdot ),\vartheta
\mathrm{grad}\text{ }(\cdot ))  \label{D4}
\end{equation}%
satisfies \cite{AM,Ar,Bl,FF,PM} \ the Jacobi identity on the space of
functionals $\mathcal{D}(M).$
\end{definition}

This means, \-in particular, that bracket (\ref{D3}) satisfies the standard
Jacobi identity on $\mathcal{D}(M).$

\begin{definition}
\label{Df_2}A linear smooth operator $\vartheta :T^{\ast }(M)\times T(M)$ is
called N\"{o}therian \cite{FF,PM,Bl} \ subject to the nonlinear dynamical
system (\ref{D1}), if the following condition
\begin{equation}
L_{K}\vartheta =\vartheta ^{^{\prime }}K-\vartheta K^{\prime ,\ast
}-K^{\prime }\vartheta =0  \label{D5}
\end{equation}%
holds identically on the manifold $M,$ where we denoted by $L_{K}$ the
corresponding Lie-derivative \cite{AM,Ar,Bl,PM} along the vector field $%
K:M\rightarrow T(M).$
\end{definition}

Assume now that the mapping $\vartheta :T^{\ast }(M)\times T(M)$ is
invertible, that is there exists the inverse mapping $\vartheta
^{-1}:=\Omega :T^{\ast }(M)\times T(M),$ and called symplectic. It then
follows easily from (\ref{D5}) then the condition
\begin{equation}
L_{K}\Omega =\Omega ^{^{\prime }}K+\Omega K^{\prime }+K^{\prime ,\ast
}\Omega =0  \label{D6}
\end{equation}%
hold identically on $M.$ Having now assumed that the manifold $M\subset
l_{2}(%
\mathbb{Z}
;%
\mathbb{C}
^{2})$ is endowed with a smooth Poissonian structure $\vartheta :T^{\ast
}(M)\times T(M)$ one can define the Hamiltonian system%
\begin{equation}
dw/dt:=-\vartheta \text{ }\mathrm{grad}\text{ }H[w],  \label{D7}
\end{equation}%
corresponding to a Hamiltonian function $H\in \mathcal{D}(M).$ As a simple
corollary of definition (\ref{D7}) one obtains that the dynamical system%
\begin{equation}
-\vartheta \text{ }\mathrm{grad}\text{ }H[w]:=K[w]  \label{D8}
\end{equation}%
satisfies the N\"{o}therian conditions (\ref{D5}). Keeping in mind the study
of the integrability problem \cite{Ar,No,BPS,Bl} subject to discrete
dynamical system (\ref{D1}), we need to construct a priori given set of
invariant with respect to it functions, \ called conservation laws, and
commuting to each other with respect to the Poisson bracket (\ref{D3}). The
following Lax criterion \cite{La,PM,BPS} proves to be very useful.

\begin{lemma}
\label{Lm_3}Any smooth solution $\varphi \in T^{\ast }(M)$ to the Lax
equation
\begin{equation}
L_{K}\text{ }\varphi =d\varphi /dt+K^{\prime ,\ast }\varphi =0,  \label{D9}
\end{equation}%
satisfying with respect to the bracket $(\cdot ,\cdot )$ the symmetry
condition
\begin{equation*}
\varphi ^{^{\prime }}=\varphi ^{^{\prime }\ast },
\end{equation*}%
is related to the conservation law
\begin{equation}
\gamma :=\underset{0}{\overset{1}{\int }}d\lambda (\varphi \lbrack w\lambda
],w).  \label{D10}
\end{equation}
\end{lemma}

\begin{proof}
The expression (\ref{D10}) easily obtains from the well-known Volterra
homology equalities:%
\begin{equation}
\gamma =\underset{0}{\overset{1}{\int }}\frac{d\gamma \lbrack w\lambda ]}{%
d\lambda }d\lambda \underset{0}{=\overset{1}{\int }}d\lambda (1,\gamma
^{\prime }[w\lambda ]\cdot w,)=\underset{0}{\overset{1}{\int }}d\lambda
(\gamma ^{\prime ,\ast }[w\lambda ]\cdot 1,w)=\underset{0}{\overset{1}{\int }%
}d\lambda (\mathrm{grad}\text{ }\gamma \lbrack w\lambda ],w)  \label{D11}
\end{equation}%
and
\begin{equation}
(\mathrm{grad}\text{ }\gamma \lbrack w])^{^{\prime }}=(\mathrm{grad}\text{ }%
\gamma \lbrack w])^{\prime ,\ast },  \label{D12}
\end{equation}%
holding identically on $M.$ \ Whence one ensues that there exists a function
$\gamma \in \mathcal{D}(M),$ such that
\begin{equation}
L_{K}\gamma =0,\text{ \ \ \ \ }\mathrm{grad}\text{ }\gamma \lbrack
w]=\varphi \lbrack w]  \label{D13}
\end{equation}%
for any $w\in M.$
\end{proof}

The Lax lemma naturally arises from the following generalized N\"{o}ther
type lemma.

\begin{lemma}
\label{Lm_4} Let a smooth element $\psi \in T^{\ast }(M)$ satisfy the N\"{o}%
ther condition
\begin{equation}
L_{K}\psi =d\psi /dt+K^{\prime ,\ast }\psi =\mathrm{grad}\text{ }\mathcal{L}%
_{\psi }  \label{D14}
\end{equation}%
\ for some smooth functional $\mathcal{L}_{\psi }\in \mathcal{D}(M).$ Then
the following Hamiltonian representation
\begin{equation}
K=-\vartheta \text{ }\mathrm{grad}\text{ }H_{\vartheta }  \label{D15}
\end{equation}%
holds, where%
\begin{equation}
\vartheta :=\psi ^{^{\prime }}-\psi ^{^{\prime }\ast }  \label{D16}
\end{equation}%
and the Hamiltonian function
\begin{equation}
\text{ }H_{\vartheta }=(\psi ,K)-\mathcal{L}_{\psi }.  \label{D16a}
\end{equation}
\end{lemma}

It is easy to see that Lemma\ \ref{Lm_3} follows from Lemma \ \ref{Lm_4}, if
the conditions $\psi ^{^{\prime }}=\psi ^{^{\prime }\ast }$ and $\mathcal{L}%
_{\psi }=0$ are imposed on \ (\ref{D14}).

Assume now that equation (\ref{D14}) allows an additional not symmetric
smooth solution $\phi \in T^{\ast }(M):$%
\begin{equation}
L_{K}\phi =d\phi /dt+K^{\prime ,\ast }\phi =\mathrm{grad}\text{ }\mathcal{L}%
_{\phi }.  \label{D17}
\end{equation}%
This means that our system (\ref{D1}) is bi-Hamiltonian:%
\begin{equation}
-\vartheta \text{ }\mathrm{grad}\text{ }H_{\vartheta }=K=-\eta \text{ }%
\mathrm{grad}\text{ }H_{\eta },  \label{D18}
\end{equation}%
where, by definition,
\begin{equation}
\eta :=\phi ^{^{\prime }}-\phi ^{^{\prime }\ast },\text{ \ }H_{\eta }=(\phi
,K)-\mathcal{L}_{\phi }.  \label{D19}
\end{equation}

\begin{definition}
\label{Df_5} One says that two Poissonian structures $\vartheta ,\eta
:T^{\ast }(M)\rightarrow T(M)$ on $M$ are compatible \cite{Ma,FF,PM,Bl}, if
for any $\lambda ,\mu \in \mathbb{R}$ the linear combination $\lambda
\vartheta +\mu \eta :T^{\ast }(M)\rightarrow T(M)$ will be also Poissonian
on $M.$
\end{definition}

It is easy to derive that this condition is satisfied if, for instance,
there exist the inverse operator $\vartheta ^{-1}:T(M)\rightarrow T^{\ast
}(M)$ and the expression $\eta (\vartheta ^{-1}\eta ):T^{\ast
}(M)\rightarrow T(M)$ is also Poissonian on $M.$

Concerning the integrability problem posed for the infinite-dimensional
dynamical system (\ref{D1}) on the discrete manifold $M$ it is, in general,
necessary, but not enough \cite{No,PM,BPS}, to prove the existence of an
infinite hierarchy of commuting to each other with respect to the Poissonian
structure (\ref{D3}) conservation laws.

Since in the case of the Lax type integrability almost of (\ref{D1}) there
exist compatible Poissonian structures and related hierarchies of
conservation laws, we will further constrain our analysis by devising an
integrability algorithm under a priori assumption that a given nonlinear
dynamical system (\ref{D1}) on the manifold $M$ is Lax type integrable. This
means that it possesses a related Lax type representation in the following,
generally written form:%
\begin{equation}
\Delta f_{n}:=f_{n+1}=l_{n}[w;\lambda ]f_{n},  \label{D20}
\end{equation}%
where $f:=\left\{ f_{n}\in
\mathbb{C}
^{r}:n\in
\mathbb{Z}
\right\} \subset l_{2}(%
\mathbb{Z}
;%
\mathbb{C}
^{r})$ for some integer $r\in
\mathbb{Z}
_{+}$ and matrices $l_{n}[w;\lambda ]\in End%
\mathbb{C}
^{r},$ $n\in
\mathbb{Z}
,$ in (\ref{D20}) are local matrix-valued functionals on $M,$ depending on
the "spectral" parameter $\lambda \in
\mathbb{C}
,$ invariant with respect to our dynamical system (\ref{D1}).

Taking into account that the Lax representation (\ref{D20}) is `local' with
respect to the discrete variable $n\in
\mathbb{Z}
,$ we will assume for convenience, that \ our manifold $M:=M_{(N)}\subset
l_{\infty }(%
\mathbb{Z}
;%
\mathbb{C}
^{m})$ is periodic with respect to the discrete index $n\in
\mathbb{Z}
_{N},$ that is for any $n\in
\mathbb{Z}
_{N}:=%
\mathbb{Z}
/N%
\mathbb{Z}
$ and $\lambda \in
\mathbb{C}
$%
\begin{equation}
l_{n}[w;\lambda ]=l_{n+N}[w;\lambda ]  \label{D21}
\end{equation}%
for some integer $N\in
\mathbb{Z}
_{+}.$ In this case the smooth functionals on $M_{(N)}$ can be represented
as
\begin{equation}
\gamma :=\sum_{n\in
\mathbb{Z}
_{N}}\gamma _{n}[w]  \label{D22}
\end{equation}%
for some local `Frechet' smooth densities $\gamma _{n}:M_{(N)}\rightarrow
\mathbb{C}
,$ $n\in
\mathbb{Z}
_{N}.$

\section{The integrability analysis: gradient-holonomic scheme}

Consider the representation (\ref{D20}) and define its fundamental solution $%
F_{m,n}(\lambda )\in Aut(%
\mathbb{C}
^{r}),$ $m,n\in
\mathbb{Z}
_{N},$ satisfying the equation
\begin{equation}
F_{m+1,n}(\lambda )=l_{m}[w;\lambda ]F_{m,n}(\lambda )  \label{D23}
\end{equation}%
and the condition
\begin{equation}
\left. F_{m,n}(\lambda )\right\vert _{m=n}=\mathbf{1}  \label{D23a}
\end{equation}%
for all $\lambda \in
\mathbb{C}
$ and $n\in
\mathbb{Z}
_{N}.$ Then the matrix function
\begin{equation}
S_{n}(\lambda ):=F_{n+N,n}(\lambda )  \label{D24}
\end{equation}%
is called the \textit{monodromy} matrix for the linear equation (\ref{D21})
and satisfies for all $n\in
\mathbb{Z}
_{N}$ \ the following Novikov-Lax type relationship:%
\begin{equation}
S_{n+1}(\lambda )l_{n}=l_{n}S_{n}(\lambda ).  \label{D25}
\end{equation}%
It easy to compute that $S_{n}(\lambda ):=\underset{k=\overleftarrow{0,N-1}}{%
\prod }l_{n+k}[w;\lambda ]$ \ \ \ owing to the periodical condition (\ref%
{D21}). \ \ Construct now the following generating functional:
\begin{equation}
\bar{\gamma}(\lambda ):=trS_{n}(\lambda )  \label{D26}
\end{equation}%
and assume that there exists its asymptotical expansion
\begin{equation}
\bar{\gamma}(\lambda )\simeq \sum_{j\in \mathbb{Z}_{+}}\bar{\gamma}%
_{j}\lambda ^{j_{0}-j}  \label{D27}
\end{equation}%
as $\lambda \rightarrow \infty $ for some fixed $j_{0}\in
\mathbb{Z}
_{+}.$ Then, owing to the evident condition%
\begin{equation}
D_{n}\overline{\gamma }(\lambda )=0  \label{D28}
\end{equation}%
for all $n\in
\mathbb{Z}
_{N},$ \ where we put, by definition, the `discrete' derivative%
\begin{equation}
D_{n}:=\Delta -1,  \label{D29}
\end{equation}%
we obtain that all functionals $\ \ \bar{\gamma}_{j}\in \mathcal{D}%
(M_{(N)}), $ $j\in
\mathbb{Z}
_{+},$ are independent of the discrete index $n\in
\mathbb{Z}
_{N}$ \ and are simultaneously conservation laws for the dynamical system (%
\ref{D1}).

Assume additionally that the following natural condition holds: the gradient
vector%
\begin{equation}
\bar{\varphi}_{n}(\lambda ):=\mathrm{grad}\text{ }\bar{\gamma}(\lambda
)_{n}[w]=tr\text{ }l^{\prime ,\ast }(S_{n}(\lambda )l_{n}^{-1}),
\label{D29a}
\end{equation}%
which solves the Lax determining equation (\ref{D9}), satisfies, owing to (%
\ref{D25}), for all $\lambda \in
\mathbb{C}
$ the next gradient relationship:%
\begin{equation}
z(\lambda )\vartheta \text{ }\bar{\varphi}(\lambda )=\eta \text{ }\bar{%
\varphi}(\lambda ),\text{ }  \label{D30}
\end{equation}%
where $z:%
\mathbb{C}
\rightarrow
\mathbb{C}
$ is some meromorphic mapping, $\vartheta $ and $\eta $ $:T^{\ast
}(M_{(N)})\rightarrow T(M_{(N)})$ are compatible Poissonian on the manifold $%
M_{(N)}$ and N\"{o}therian operators with respect to the dynamical system (%
\ref{D1}). \ As a corollary of the above condition one follows easily that
the generating functional $\overline{\gamma }(\lambda )\in \mathcal{D}%
(M_{(N)})$ satisfies the commutation relationships
\begin{equation}
\{\bar{\gamma}(\lambda ),\bar{\gamma}(\mu )\}_{\vartheta }=0=\{\bar{\gamma}%
(\lambda ),\bar{\gamma}(\mu )\}_{\eta }  \label{D30a}
\end{equation}%
for all $\lambda ,\mu \in \mathbb{C}.$ \ Thereby, if to define on the
manifold $M_{(N)}$ \ a generating \ dynamical system
\begin{equation}
dw/d\tau :=-\vartheta \text{ }\mathrm{grad}\text{ }\bar{\gamma}(\lambda )[w]
\label{D30b}
\end{equation}%
\ \ as $\lambda \rightarrow \infty ,$ it follows easily from (\ref{D30a})
that the hierarchy of functionals (\ref{D27}) is its conservation law.

Since the existence of an infinite hierarchy of commuting to each other
conservation laws is characteristic concerning the Lax type integrability of
the nonlinear dynamical system (\ref{D1}), this property can be effectively
implemented into the scheme of our analysis. Namely, the following statement
holds.

\begin{proposition}
\label{Pr_5} The Lax equation (\ref{D9}) allows the following asymptotical
as $\lambda \rightarrow \infty $ periodical solution $\varphi (\lambda )\in
T^{\ast }(M_{(N)}):$
\begin{equation}
\varphi _{n}(\lambda )\simeq a_{n}(\lambda )\exp [\omega (t;\lambda )]%
\underset{j=0}{\overset{n}{\prod }}\sigma _{j}(\lambda ),  \label{D31}
\end{equation}%
where for all $n\in \mathbb{Z}$
\begin{eqnarray}
a_{n}(\lambda ) &:&=(1,a_{(1),n}[w;\lambda ],a_{(2),n}[w;\lambda
],...,a_{(m-1),n}[w;\lambda ])^{\tau },  \label{D32} \\
a_{(k),n}(\lambda ) &\simeq &\underset{s\in
\mathbb{Z}
_{+}}{\sum }a_{(k),n}^{(s)}[w]\lambda ^{-s+\widetilde{a}},\sigma
_{j}(\lambda )\simeq \underset{s\in
\mathbb{Z}
_{+}}{\sum }a_{j}^{(s)}[w]\lambda ^{-s+\widetilde{\sigma }},  \notag
\end{eqnarray}%
$k=\overline{1,m-1}$ and $\omega (t;\cdot ):\mathbb{C}\rightarrow \mathbb{C}%
, $ $t\in \mathbb{R},$ is some dispersion function. Moreover the functional $%
\gamma (\lambda ):=\underset{n\in \mathbb{Z}_{N}}{\sum }\ln (\lambda ^{-%
\tilde{\sigma}}\sigma _{n}[w;\lambda ])\in \mathcal{D}(M_{(N)})$ is a
generating function of conservation laws for the dynamical system (\ref{D1}).
\end{proposition}

\begin{proof}
Owing to Lemma\ (\ref{Lm_3}) and relationship (\ref{D29a}) functional (\ref%
{D26}) is a conservation law for our dynamical system (\ref{D1}). Based now
on expression (\ref{D24}) and equation (\ref{D20}) one arrives at the
solution representation (\ref{D31}) to the Lax equation (\ref{D9}). Now,
making use of the periodicity of the manifold \ $M_{(N)},$ we obtain \ from
the period translation of (\ref{D31}) \ that the functional
\begin{equation}
\gamma (\lambda ):=\underset{n\in \mathbb{Z}_{N}}{\sum }\ln (\lambda ^{-%
\tilde{\sigma}}\sigma _{n}[w;\lambda ])\simeq \underset{j\in \mathbb{Z}_{+}}{%
\sum }\gamma _{j}\lambda ^{-j}  \label{D33}
\end{equation}%
generates an infinite hierarchy of conservation laws to (\ref{D1}), \ so
finishing the proof.
\end{proof}

Thus, if \ we start the Lax type integrability analysis of a priori given
nonlinear dynamical system (\ref{D1}), it is necessary, as the first step,
to study the asymptotical solutions (\ref{D31}) to the corresponding Lax
equation (\ref{D9}) and construct a related hierarchy of conservation laws
in the functional form (\ref{D33}), taking into account expansions (\ref{D32}%
).

\begin{remark}
\label{Rm_6} It is easy to observe that, owing to the arbitrariness of the
period $N\in
\mathbb{Z}
_{+}$ of the manifold $M_{(N)},$ all of the finite-sum expressions obtained
above can be generalized to the corresponding infinite dimensional manifold $%
M\subset l_{2}(%
\mathbb{Z}
;%
\mathbb{C}
^{m}),$ if the corresponding infinite series persist to be convergent.
\end{remark}

Since our dynamical system (\ref{D1}) under the above conditions is a
bi-Hamiltonian flow on the manifold $M_{(N)},$ as the next step of the
related compatible Poissonian or symplectic structures, satisfying,
respectively, either equality (\ref{D5}) or equality \ (\ref{D6}). Before
doing this, we need to formulate the following lemma.

\begin{lemma}
\label{Lm_7} All functionals $\gamma _{j}\in \mathcal{D}(M_{(N)}),$ entering
the expansion (\ref{D33}), are commuting to each other with respect to both
Poissonian structures $\vartheta ,\eta :T^{\ast }(M_{(N)})\rightarrow
T(M_{(N)}),$ satisfying the gradient relationship(\ref{D35}).

\begin{proof}
Based on representations (\ref{D31}) and (\ref{D29a}) one obtains that there
holds the asymptotical as $\lambda \rightarrow \infty $ \ relationship%
\begin{equation}
\ln \bar{\gamma}(\lambda )\simeq \gamma (\lambda ).  \label{D35}
\end{equation}%
Since the generative function $\bar{\gamma}(\lambda )\in \mathcal{D}%
(M_{(N)}) $ satisfies the commutation relationships (\ref{D30a}), the same
also holds, owing to (\ref{D35}), for the generating function $\gamma
(\lambda )\in \mathcal{D}(M_{(N)}),$ finishing the proof.
\end{proof}
\end{lemma}

Proceed now to constructing the related to dynamical system (\ref{D1})
Poissonian structures $\vartheta ,\eta :T^{\ast }(M_{(N)})\rightarrow
T(M_{(N)}).$ \ Note here, that these Poissonian structures are N\"{o}therian
also for the whole hierarchy of dynamical systems
\begin{equation}
dw/dt_{j}:=-\vartheta \text{ }\mathrm{grad}\text{ }\gamma _{j}[w],
\label{D36}
\end{equation}%
where $t_{j}\in \mathbb{R},j\in
\mathbb{Z}
_{+},$ are the corresponding\ evolution parameters, and which, owing to (\ref%
{D30a}), commute to each other on the manifold $M_{(N)}.$ The latter makes
it possible to apply Lemma \ \ref{Lm_4} to arbitrary one of the dynamical
systems (\ref{D36}), if the related vector fields commuting with (\ref{D1})
are supposed to be found before.

To solve analytically equation (\ref{D14}) subject to an element $\varphi
\in T^{\ast }(M_{(N)})$ one can, in the case of a polynomial dynamical
system (\ref{D1}), make use of the well known asymptotical small parameter
method \cite{MBPS,PM}. If applying this approach, it is necessary to take
into account the following expansions at zero - element $\overline{w}=0\in
M_{(N)}$ with respect to the small parameter $\mu \rightarrow 0:$%
\begin{eqnarray}
w &:&=\mu w^{(1)},\text{ }\varphi \lbrack w]=\varphi ^{(0)}+\mu \varphi
^{(1)}[w]+\mu ^{2}\varphi ^{(2)}[w]+...,  \notag \\
d/dt &=&d/dt_{0}+\mu d/dt_{1}+\mu ^{2}d/dt_{2}+...,  \notag \\
K[w] &=&\mu K^{(1)}[w]+\mu ^{(2)}K^{(2)}[w]+...,  \label{D37} \\
K^{\prime }[w] &=&K_{0}^{\prime }+\mu K_{1}^{\prime }[w]+\mu
^{2}K_{2}^{\prime }[w]+...,  \notag \\
\mathrm{grad}\text{ }\mathcal{L}[w] &=&\mathrm{grad}\text{ }\mathcal{L}%
^{(0)}+\mu \mathrm{grad}\text{ }\mathcal{L}^{(1)}[w]+\mu ^{2}\mathrm{grad}%
\text{ }\mathcal{L}^{(2)}[w]+...\text{ }.  \notag
\end{eqnarray}%
Having solved the corresponding set of linear nonuniform functional equations%
\begin{eqnarray}
d\varphi ^{(0)}/dt_{0}+K_{0}^{\prime \ast }\varphi ^{(0)} &=&\mathrm{grad}%
\text{ }\mathcal{L}^{(0)},  \notag \\
d\varphi ^{(1)}/dt_{0}+K_{0}^{\prime \ast }\varphi ^{(1)} &=&\mathrm{grad}%
\text{ }\mathcal{L}^{(1)}-K_{0}^{\prime \ast }\varphi ^{(0)},  \label{D38} \\
d\varphi ^{(2)}/dt_{0}+K_{0}^{\prime \ast }\varphi ^{(2)} &=&\mathrm{grad}%
\text{ }\mathcal{L}^{(2)}-K_{1}^{\prime \ast }\varphi ^{(1)}-K_{2}^{\prime
\ast }\varphi ^{(0)}  \notag
\end{eqnarray}%
and so on, by means the standard Fourier transform applied to the suitable $%
N $-periodical functions, one can obtain the related Poissonian structure as
the series
\begin{equation}
\vartheta ^{-1}=\varphi ^{(0),\prime }-\varphi ^{(0),\prime \ast }+\mu
(\varphi ^{(1),\prime }-\varphi ^{(1),\prime \ast })+...  \label{D39}
\end{equation}%
and next to put in (\ref{D39}) at the end $\mu =1.$

Another direct way of obtaining an Poissonian operator $\vartheta :T^{\ast
}(M_{(N)})\rightarrow T(M_{(N)})$ for (\ref{D1}) is to solve by means of the
same asymptotical small parameter approach the N\"{o}therian equation (\ref%
{D5}), having preliminary reduced it to the following set of linear
nonuniform equations:%
\begin{eqnarray}
\frac{d}{dt_{0}}(\vartheta _{0}\varphi ^{(0)}) &=&K_{0}^{\prime }(\vartheta
_{0}\varphi ^{(0)}),  \label{D40} \\
\frac{d}{dt_{0}}(\vartheta _{1}\varphi ^{(0)}) &=&K_{0}^{\prime }(\vartheta
_{1}\varphi ^{(0)})+\vartheta _{0}K_{1}^{\prime ,\ast }\varphi
^{(0)}+K_{1}^{\prime }\vartheta _{0}\varphi ^{(0)},  \notag \\
\frac{d}{dt_{0}}(\vartheta _{2}\varphi ^{(0)}) &=&K_{0}^{\prime }(\vartheta
_{2}\varphi ^{(0)})-\varphi ^{(0)\prime }K^{1}+\vartheta _{0}K_{2}^{\prime
,\ast }\varphi ^{(0)}+  \notag \\
&&+\vartheta _{1}K_{1}^{\prime ,\ast }\varphi ^{(0)}+\vartheta
_{2}K_{0}^{\prime ,\ast }\varphi ^{(0)}+K_{1}^{\prime }\vartheta _{1}\varphi
^{(0)}+K_{2}^{\prime }\vartheta _{0}\varphi ^{(0)}.  \notag
\end{eqnarray}%
Based now on the analytical expressions for actions $\vartheta _{j}:\varphi
^{(0)}\rightarrow \vartheta _{j}\varphi ^{(0)},$ $j\in
\mathbb{Z}
_{+},$ one can easily retrieve them in operator form from the expansion
\begin{equation}
\vartheta =\vartheta _{0}+\mu \vartheta _{1}+\mu ^{2}\vartheta _{2}+...,
\label{D41}
\end{equation}%
if to put at the end of calculations $\mu =1.$ Similarly one can also
construct the second Poissonian operator $\eta :T^{\ast
}(M_{(N)})\rightarrow T(M_{(N)})$ for the nonlinear dynamical system (\ref%
{D1}).

Resuming up all this analysis described above, we can formulate the
following proposition.

\begin{proposition}
\label{Pr_7} Let a nonlinear dynamical system (\ref{D1}) on a discrete
manifold $M_{(N)}$ allow both a nontrivial symmetric solution $\varphi \in $
$T^{\ast }(M_{(N)})$ to the Lax equation (\ref{D9}) in the asymptotical as $%
\lambda \rightarrow \infty $ \ form (\ref{D31}), generating an infinite
hierarchy of nontrivial functionally independent conservation laws (\ref{D33}%
), and compatible nonsymmetric solutions $\psi $ and $\ \phi \in T^{\ast
}(M_{(N)})$ to the N\"{o}ther equations (\ref{D14} and (\ref{D17}),
respectively. Then this dynamical system is a Lax type integrable
bi-Hamiltonian flow on $M_{(N)}$ with respect to two compatible Poissonian
structures $\vartheta ,\eta :T^{\ast }(M_{(N)})\rightarrow T(M_{(N)}),$
whose adjoint Lax type representation%
\begin{equation}
d\Lambda /dt=[\Lambda ,K^{\prime ,\ast }],  \label{D42}
\end{equation}%
where $\Lambda :=\vartheta ^{-1}\eta $ is the so-called recursion operator,
can be transformed, owing to the gradient relationship (\ref{D30}), to the
standard discrete Lax type form
\begin{equation}
dl_{n}/dt=[p_{n}(l),l_{n}]+(D_{n}p_{n}(l))l_{n}  \label{D43}
\end{equation}%
for some matrix $\ p_{n}(l)\in End%
\mathbb{C}
^{r}$ describing the related to (\ref{D20}) temporal evolution
\begin{equation}
df_{n}/dt=p_{n}(l)f_{n}  \label{D44}
\end{equation}%
for $f\in l_{\infty }(\mathbb{Z};%
\mathbb{C}
^{r}).$
\end{proposition}

\begin{remark}
\label{Re_8} Based on the property that all Hamiltonian flows (\ref{D36})
commute to each other and to dynamical system (\ref{D1}) and using the fact
that they possess \ the same Poissonian and compatible $(\vartheta ,\eta )$%
-pair, the analytical algorithm described above can be equally applied to
any other flow, commuting with (\ref{D1}).
\end{remark}

Concerning the discrete linear Lax type problem (\ref{D20}), it can be
constructed by means of the gradient-holonomic algorithm, devised in \cite%
{PM,HPP,BPS} for studying the integrability of nonlinear dynamical systems
on functional manifolds. Namely, making use of the preliminary found
analytical expressions for the related compatible Poissonian structures $%
\vartheta ,\eta :T^{\ast }(M_{(N)})\rightarrow T(M_{(N)})$ on the manifold $%
M_{(N)}$ and using the fact that the recursion operator $\Lambda :=\vartheta
^{-1}\eta :$ $T^{\ast }(M_{(N)})\rightarrow T^{\ast }(M_{(N)})$ satisfies
the dual Lax type commutator equality (\ref{D42}), one can retrieve the
standard Lax type representation for \ it by means of suitably derived
algebraic relationships. As a corollary of Proposition \ \ref{Pr_7} \ one
can also claim that the existence of a nontrivial asymptotical as $\lambda
\rightarrow \infty $ \ solution to the Lax equation (\ref{D9}) can serve as
an effective Lax type integrability criterion for a given nonlinear
dynamical system (\ref{D1}) on the manifold $M_{(N)}.$

\section{ The Bogoyavlensky-Novikov finite-dimensional reduction}

Assume that our dynamical system (\ref{D1}) on the periodic manifold $%
M_{(N)} $ is Lax type integrable and possesses two compatible Poissonian
structures $\vartheta ,\eta :T^{\ast }(M_{(N)})\rightarrow T(M_{(N)}).$Thus,
we have the nonlinear finite-dimensional dynamical system%
\begin{equation}
dw_{n}/dt:=K_{n}[w]=-\vartheta \text{ }\mathrm{grad}\text{ }H_{n}[w]
\label{D45}
\end{equation}%
for indices $n\in
\mathbb{Z}
_{N}$ owing to its $N$-periodicity. The finite dimensional dynamical system (%
\ref{D45}) can be equivalently considered as that on the finite-dimensional
space $M_{(N)}\simeq (%
\mathbb{C}
^{m})^{N}$ parameterized by any integer index $n\in
\mathbb{Z}
_{N},$ and whose Liouville integrability is of the next our analysis. To
proceed with studying the flow (\ref{D45}) on the manifold $M_{(N)},$ we
will make use of the Bogoyavlensky-Novikov \cite{BN,No} reduction scheme
\cite{No,Pr,PM,Bl}.

Let $\Lambda (M_{(N)}):=\underset{j\in
\mathbb{Z}
_{+}}{\oplus }\Lambda ^{j}(M_{(N)})$ be the standard finitely generated
Grassmann algebra \cite{Ar,PM,BPS} of differential forms on the manifold $%
M_{(N)}.$ Then the following differential complex
\begin{equation}
\Lambda ^{0}(M_{(N)})\overset{d}{\rightarrow }\Lambda ^{1}(M_{(N)})\overset{d%
}{\rightarrow ...}\overset{d}{\rightarrow }\Lambda ^{j}(M_{(N)})\overset{d}{%
\rightarrow }\Lambda ^{j+1}(M_{(N)})\overset{d}{\rightarrow }....,
\label{D46}
\end{equation}%
where $d:\Lambda (M_{(N)})\rightarrow \Lambda (M_{(N)})$ is the external
differentiation, is \ finite and exact. \ Since the discrete `derivative' $%
D_{n}:=\Delta -1$ commutes with the differentiation $d:\Lambda
(M_{(N)})\rightarrow \Lambda (M_{(N)}),$ $[D_{n},d]=0$ for all $n\in
\mathbb{Z}
_{N},$ and for any element $a\in \Lambda ^{0}(M_{(N)})$
\begin{equation}
\mathrm{grad}(\sum_{n\in
\mathbb{Z}
_{N}}D_{n}a_{n}[w])=0,  \label{D47}
\end{equation}%
one can formulate the following Gelfand-Dikiy type \cite{GD} lemma.

\begin{lemma}
\label{Lm_8} Let $\mathcal{L}[w]\in \Lambda ^{0}(M_{(N)})$ be a Frechet
smooth local Lagrangian functional on the manifold $M_{(N)}.$ Then there
exists a differential 1-form $\alpha ^{(1)}\in \Lambda ^{1}(M_{(N)}),$ such
that the equality%
\begin{equation}
d\mathcal{L}_{n}[w]=<\mathrm{grad}\text{ }\mathcal{L}_{n}[w],dw_{n}>+D_{n}%
\alpha _{n}^{(1)}[w]  \label{D48}
\end{equation}%
holds for all $n\in
\mathbb{Z}
_{N}.$
\end{lemma}

\begin{proof}
One can easily observe that
\begin{eqnarray}
d\mathcal{L}_{n}[w] &=&\underset{j=0}{\overset{N-1}{\sum }}<\frac{\partial
\mathcal{L}_{n}[w]}{\partial w_{n+j}},dw_{n+j}>=\underset{j=0}{\overset{N-1}{%
\sum }}<\frac{\partial \mathcal{L}_{n}[w]}{\partial w_{n+j}},\Delta
^{j}dw_{n}>=  \label{D49} \\
&=&<\underset{j=0}{\overset{N-1}{\sum }}\Delta ^{-j}\frac{\partial \mathcal{L%
}_{n}[w]}{\partial w_{n+j}},dw_{n}>+D_{n}\left( \underset{j=0}{\overset{N-1}{%
\sum }}<p_{j},dw_{n+j}>\right) ,  \notag
\end{eqnarray}%
where, by definition,
\begin{equation}
p_{k}:=\underset{j=0}{\overset{N-1}{\sum }}\Delta ^{-j}\frac{\partial
\mathcal{L}_{n}[w]}{\partial w_{n+j+k+1}}  \label{D50}
\end{equation}%
for $k=\overline{0,N-1}.$ \ Having denoted the expression
\begin{equation}
\mathrm{grad}\text{ }\mathcal{L}_{n}[w]:=\underset{j=0}{\overset{N-1}{\sum }}%
\Delta ^{-j}\frac{\partial \mathcal{L}_{n}[w]}{\partial w_{n+j}},
\label{D51}
\end{equation}%
one obtains the result (\ref{D48}), where
\begin{equation}
\alpha _{n}^{(1)}[w]:=\underset{j=0}{\overset{N-1}{\sum }}<p_{j},dw_{n+j}>
\label{D52}
\end{equation}%
is the corresponding differential 1-form on the manifold $M_{(N)}.$
\end{proof}

Applying now the $d$-differentiation to expression (\ref{D48}) we obtain
that
\begin{equation}
-D_{n}\omega _{n}^{(2)}[w]=<d\text{ }\mathrm{grad}\text{ }\mathcal{L}%
_{n}[w],\wedge dw_{n}>  \label{D53}
\end{equation}%
for any $n\in
\mathbb{Z}
,$ where the 2-form
\begin{equation}
\omega ^{(2)}[w]:=d\alpha ^{(1)}[w]  \label{D54}
\end{equation}%
is nondegenerate on $M_{(N)},$ if the Hessian $\partial _{n}^{2}\mathcal{L}%
[w]/\partial ^{2}w_{n}$ is also nondegenerate.

Assume now that the submanifold
\begin{equation}
\bar{M}_{(N)}:=\left\{ \mathrm{grad}\ \mathcal{L}_{n}^{(\bar{N})}[w]=0;\text{
}w\in M_{(N)}\right\} ,  \label{D55}
\end{equation}%
where, by definition, the Lagrangian functional
\begin{equation}
\mathcal{L}^{(\bar{N})}:=-\gamma _{\bar{N}}+\underset{j=0}{\overset{\bar{N}-1%
}{\sum }}c_{j}\gamma _{j},  \label{D56}
\end{equation}%
with $\gamma _{j}\in $ $\mathcal{D}(M),$ $j=\overline{0,\bar{N}-1},$ for
some $\bar{N}\in \mathbb{Z}_{+},$ being the suitable nontrivial conservation
laws for the dynamical system (\ref{D1}), constructed before, and $c_{j}\in
\mathbb{C},$ $j=\overline{0,\bar{N}-1},$ being some arbitrary but fixed
constants. As a result of (\ref{D55}) and (\ref{D53}) we obtain that closed
2-form $\omega ^{(2)}\in \Lambda ^{2}(\bar{M}_{(N)})$ is invariant with
respect to the index $n\in
\mathbb{Z}
_{N}$ on the manifold $\bar{M}_{(N)}.$ Moreover, the submanifold (\ref{D55})
is also invariant both with respect to the index $n\in \mathbb{%
\mathbb{Z}
}_{N}$ and the evolution parameter $t\in \mathbb{R}.$ Really, for any $n\in
\mathbb{%
\mathbb{Z}
}_{N}$ the Lie derivative
\begin{equation}
L_{K}\mathrm{grad}\ \mathcal{L}^{(\bar{N})}=(\mathrm{grad}\ \mathcal{L}^{(%
\bar{N})})^{\prime }K+K^{\prime ,\ast }(\mathrm{grad}\ \mathcal{L}^{(\bar{N}%
)})=0,  \label{D57}
\end{equation}%
since the functional $\mathcal{L}^{(\bar{N})}\in \mathcal{D}(M_{(N)})$ is a
sum of conservation laws for the dynamical system (\ref{D1}), \ whose
gradients satisfies the Lax condition (\ref{D9}). \ Moreover, it is easy to
see that if the Lie derivative $L_{K}$ $\mathrm{grad}$ $\mathcal{L}_{n}^{(%
\bar{N})}[w]=0,n\in \mathbb{%
\mathbb{Z}
}_{N},$ at $t=0,$ then $\mathrm{grad}\ \mathcal{L}_{n}^{(\bar{N})}[w]=0$ for
all $t\in \mathbb{R}$ and $n\in \mathbb{%
\mathbb{Z}
}_{N}.$ Thus, the Bogoyavlensky-Novikov reduction of the dynamical system (%
\ref{D1}) upon the invariant submanifold $\bar{M}_{(N)}$ \ is defined
completely invariantly.

Now a question arises: \ how are related the dynamical system (\ref{D1}),
naturally constrained to live on the submanifold $M_{(N)},$ and the
dynamical system (\ref{D1}), reduced on the finite dimensional submanifold $%
\bar{M}_{(N)}\subset M_{(N)}.$ To analyze this reduction we will consider
the following equality:%
\begin{equation}
<\mathrm{grad}\ \mathcal{L}_{n}^{(\bar{N}%
)}[w],K_{n}[w]>=-D_{n}h_{n}^{(t)}[w],  \label{D59}
\end{equation}%
for some local functional $h^{(t)}[w]\in \Lambda ^{0}(M),$ following from
the conditions (\ref{D47}) and (\ref{D9}):%
\begin{eqnarray}
\mathrm{grad}\ &<&\mathrm{grad}\ \mathcal{L}_{n}^{(\bar{N})}[w],K_{n}[w]>=(%
\mathrm{grad}\ \mathcal{L}_{n}^{(\bar{N})}[w])^{\prime ,\ast
}K_{n}[w]+K_{n}^{\prime ,\ast }[w]\mathrm{grad}\ \mathcal{L}_{n}^{(\bar{N}%
)}[w]=  \label{D60} \\
&=&(\mathrm{grad}\ \mathcal{L}_{n}^{(\bar{N})}[w])^{\prime
}K_{n}[w]+K_{n}^{\prime ,\ast }[w]\mathrm{grad}\ \mathcal{L}_{n}^{(\bar{N}%
)}[w]=L_{K}\mathrm{grad}\ \mathcal{L}_{n}^{(\bar{N})}[w]=0,  \notag
\end{eqnarray}%
giving rise to (\ref{D59}). \ Since on the submanifold $\bar{M}_{(N)}$ the
gradient $\mathrm{grad}$ $\mathcal{L}_{n}^{(\bar{N})}[w]=0$ for all $n\in
\mathbb{%
\mathbb{Z}
}_{N},$ we obtain from (\ref{D59}) that the local functional $h^{(t)}[w]\in
\Lambda ^{0}(\bar{M}_{(N)})$ does not depend on index $n\in \mathbb{%
\mathbb{Z}
}_{N}.$

The properties of the manifold $\bar{M}_{(N)},$ described above, make it
possible to consider it as a symplectic manifold endowed with the symplectic
structure $\omega ^{(2)}\in \Lambda ^{2}(\bar{M}_{(N)}),$ given by
expressions (\ref{D52}) and (\ref{D54}). From this point of view we can
proceed to studying the integrability properties of the dynamical system (%
\ref{D1}) reduced upon the invariant finite-dimensional manifold $\bar{M}%
_{(N)}\subset M_{(N)}.$

First, we observe that the vector field $d/dt$ on $\bar{M}_{(N)}$ \ is
canonically Hamiltonian \cite{AM,Ar,No} with respect to the symplectic
structure $\omega ^{(2)}\in \Lambda ^{2}(\bar{M}_{(N)}):$%
\begin{equation}
-i_{\frac{d}{dt}}\omega ^{(2)}(w,p)=dh^{(t)}(w,p),  \label{D61}
\end{equation}%
where $h^{(t)}(w,p):=$ $h^{(t)}(w),\omega ^{(2)}(w,p):=\omega ^{(2)}[w]$ and
$(w,p)^{\intercal }\in \bar{M}_{(N)}$ are canonical variables induced on the
manifold $\bar{M}_{(N)}$ by the Liouville 1-form (\ref{D52}). \ Really, from
expression (\ref{D59}) one obtains that
\begin{equation*}
di_{\frac{d}{dt}}<\mathrm{grad}\ \mathcal{L}_{n}^{(\bar{N}%
)}[w],dw_{n}>=-D_{n}dh_{n}^{(t)}[w],
\end{equation*}%
which, being supplemented with the identity (\ref{D53}) in the form
\begin{equation*}
i_{\frac{d}{dt}}d<\mathrm{grad}\ \mathcal{L}_{n}^{(\bar{N}%
)}[w],dw_{n}>=-D_{n}i_{\frac{d}{dt}}\omega _{n}^{(2)}[w],
\end{equation*}%
entails the following:

\begin{equation}
\frac{d}{dt}<\mathrm{grad}\ \mathcal{L}_{n}^{(\bar{N}%
)}[w],dw_{n}>=-D_{n}(dh_{n}^{(t)}[w]+i_{\frac{d}{dt}}\omega _{n}^{(2)}[w]),
\label{D62}
\end{equation}%
Since $\mathrm{grad}$ $\mathcal{L}^{(\bar{N})}[w]=0=L_{K}$ $\mathrm{grad}$ $%
\mathcal{L}[w]$ \ identically \ \ on $\ \bar{M}_{(N)},$ \ from (\ref{D62})
one obtains \ the result (\ref{D61}).

The same one can claim subject to any of Hamiltonian systems (\ref{D36}),
commuting with (\ref{D1}) on the manifold $M.$ Moreover, owing to the
functional independence of invariants $\gamma _{j}\in $ $\mathcal{D}%
(M_{(N)}),$ $j=\overline{0,\bar{N}-1},$ entering the Lagrangian functional (%
\ref{D56}), we can construct the set of functionally independent functions $%
h^{(j)}\in \mathcal{D}(\bar{M}_{(N)}),$ $j=\overline{0,\bar{N}-1},$ as
follows:
\begin{equation}
<\mathrm{grad}\ \mathcal{L}_{n}^{(\bar{N})}[w],\vartheta _{n}\text{ }\mathrm{%
grad}\ \mathcal{\gamma }_{j,n}[w]>=D_{n}h_{n}^{(j)}[w],  \label{D63}
\end{equation}%
It is easy to check that these functions $h^{(j)}\in \mathcal{D}(\bar{M}%
_{(N)}),j=\overline{0,\bar{N}-1},$ are invariant with respect to indices $%
n\in \mathbb{Z}_{N}$ and commuting both to each other and to the Hamiltonian
function $h^{(t)}\in \mathcal{D}(\bar{M}_{(N)})$ with respect to the
symplectic structure $\omega ^{(2)}\in \Lambda ^{2}(\bar{M}_{(N)}).$ Thus,
if the dimension $\dim \bar{M}_{(N)}=2\bar{N},$ the discrete dynamical
system \ (\ref{D1}) reduced upon the finite-dimensional submanifold $\bar{M}%
_{(N)}\subset M_{(N)}$ will be Liouville integrable. If the set of
conservation laws $\gamma _{j}\in \mathcal{D}(M_{(N)}),$ $j=\overline{0,\bar{%
N}-1},$ proves to be functionally dependent on $M_{(N)},$ the described
scheme should be modified by means of using the Dirac reduction technique
\cite{AM,Bl,PM} for regular finding the symplectic structure $\bar{\omega}%
^{(2)}[w]\in \Lambda ^{2}(\bar{M}_{(N)})$ on invariant nonsingular
submanifolds.

\section{Example 1: the differential-difference nonlinear Schr\"{o}dinger
dynamical system and its integrability}

The mentioned before discrete nonlinear Schr\"{o}dinger dynamical system \ (%
\ref{D2}) is defined on the periodic manifold $M_{(N)}\subset l_{\infty }(%
\mathbb{Z};\mathbb{C}^{2}).$ Its Lax type integrability was stated in \cite%
{AL,MS,BP} making use of the simplest discretization of the standard
Zakharov-Shabat spectral problem for the well-known nonlinear Schr\"{o}%
dinger equation. In this Section \ we will demonstrate the application of
the gradient-holonomic integrability analysis, described above, to this
discrete nonlinear Schr\"{o}dinger dynamical system \ (\ref{D2}). First, we
will show the existence of an infinite hierarchy of functionally independent
conservation laws, having solved the determining Lax equation (\ref{D9}) in
the asymptotical form (\ref{D31}). The following lemma holds.

\begin{lemma}
\label{Lm_8a}\ \ The functional expression%
\begin{equation}
\varphi _{n}:=\binom{1}{a_{n}(\lambda )}\exp [it(2-\lambda -\lambda
^{-1})]\prod\limits_{j=1}^{n}\sigma _{j}(\lambda ),  \label{D64}
\end{equation}%
is an asymptotical, as $\ \ \lambda \rightarrow \infty ,$ \ solution to the
determining Lax equation
\begin{equation}
d\varphi _{n}/dt+K_{n}^{\prime ,\ast }\varphi _{n}=0  \label{D65a}
\end{equation}%
for all $n\in \mathbb{Z}_{N}$ with the operator $K^{\prime ,\ast }:T^{\ast
}(M_{(N)})\rightarrow T^{\ast }(M_{(N)})$ of the form:
\begin{equation}
K_{n}^{\prime ,\ast }=\left(
\begin{array}{cc}
\begin{array}{c}
i\Delta ^{-1}D_{n}^{2}-iv_{n}(u_{n+1}+u_{n-1})- \\
-i(\Delta +\Delta ^{-1})\cdot v_{n}u_{n}%
\end{array}
& iv_{n}(v_{n+1}+v_{n-1}) \\
-iu_{n}(u_{n+1}+u_{n-1}) &
\begin{array}{c}
-i\Delta ^{-1}D_{n}^{2}+iu_{n}(v_{n+1}+v_{n-1})+ \\
+i(\Delta +\Delta ^{-1})\cdot v_{n}u_{n}%
\end{array}%
\end{array}%
\right) ,  \label{D66}
\end{equation}%
where, by definition,%
\begin{eqnarray}
\sigma _{n}(\lambda ) &\simeq &\frac{\lambda }{h_{n}[u,v]}%
(1-\sum\limits_{s\in \mathbb{Z}_{+}}\sigma _{n+}^{(s)}[u,v]\lambda ^{-s-1}),
\label{D65} \\
a_{n}(\lambda ) &\simeq &\sum\limits_{s\in \mathbb{Z}_{+}}a_{n}^{(s)}[u,v]%
\lambda ^{-s}  \notag
\end{eqnarray}%
are the corresponding asymptotical expansions.
\end{lemma}

\begin{proof}
To prove \ this Lemma it is enough to find the corresponding coefficients \
of the asymptotical expansions \ (\ref{D65}). To do \ this we will consider
the following two equations easily obtained from \ (\ref{D65a}), (\ref{D66})
and (\ref{D64}) :%
\begin{equation}
\begin{array}{c}
D_{n}^{-1}\frac{d}{dt}[-\ln h_{n}+\ln (1-\sum\limits_{s\in \mathbb{Z}%
_{+}}\sigma _{n}^{(s)}\lambda ^{-s-1})]+ \\
+i\lambda \lbrack h_{n+1}^{-1}(1-v_{n}u_{n})(1-\sum\limits_{s\in \mathbb{Z}%
_{+}}\sigma _{n}^{(s)}\lambda ^{-s-1})-1]+ \\
+\frac{i}{\lambda }\left[ (1-v_{n-1}u_{n-1})h_{n}(1-\sum\limits_{s\in
\mathbb{Z}_{+}}\sigma _{n}^{(s)}\lambda ^{-s-1})^{-1}-1\right] - \\
-iv_{n}(u_{n+1}+u_{n-1})+iv_{n}(v_{n+1}+v_{n-1})\sum\limits_{s\in \mathbb{Z}%
_{+}}a_{n}^{(s)}\lambda ^{-s}%
\end{array}
\label{D67}
\end{equation}%
and
\begin{equation}
\begin{array}{c}
(\sum\limits_{s\in \mathbb{Z}_{+}}a_{n}^{(s)}\lambda ^{-s})D_{n}^{-1}\frac{d%
}{dt}[-\ln h_{n}+\ln (1-\sum\limits_{s\in \mathbb{Z}_{+}}\sigma
_{n}^{(s)}\lambda ^{-s-1})]+4i(\sum\limits_{s\in \mathbb{Z}%
_{+}}a_{n}^{(s)}\lambda ^{-s})+ \\
+\left[ i\lambda h_{n+1}(v_{n+1}u_{n+1}-1)(\sum\limits_{s\in \mathbb{Z}%
_{+}}a_{n+1}^{(s)}\lambda ^{-s})(\sum\limits_{s\in \mathbb{Z}%
_{+}}a_{n+1}^{(s)}\lambda ^{-s})-\sum\limits_{s\in \mathbb{Z}%
_{+}}a_{n}^{(s)}\lambda ^{-s}\right] + \\
+\frac{i}{\lambda }\left[ (v_{n-1}u_{n-1}-1)(\sum\limits_{s\in \mathbb{Z}%
_{+}}a_{n+1}^{(s)}\lambda ^{-s})h_{n}(1-\sum\limits_{s\in \mathbb{Z}%
_{+}}\sigma _{n}^{(s)}\lambda ^{-s-1})^{-1}-\sum\limits_{s\in \mathbb{Z}%
_{+}}a_{n}^{(s)}\lambda ^{-s}\right] + \\
+\frac{d}{dt}\sum\limits_{s\in \mathbb{Z}_{+}}a_{n}^{(s)}\lambda
^{-s}-iu_{n}(u_{n+1}+u_{n-1})+iu_{n}(v_{n+1}+v_{n-1})\sum\limits_{s\in
\mathbb{Z}_{+}}a_{n}^{(s)}\lambda ^{-s}.%
\end{array}
\label{D67a}
\end{equation}%
Having equated the coefficients of \ (\ref{D67}) at the same degrees of the
parameter $\lambda \in \mathbb{C},$ we obtain step-by-step the functional
expressions for $h_{n},\sigma _{n}^{(s)}$ and $a_{n}^{(s)}$%
\begin{eqnarray}
h_{n} &=&(1-v_{n}u_{n}),a_{n}^{(0)}=0,a_{n}^{(1)}=\beta ,  \label{D68} \\
\sigma _{n}^{(0)}
&=&v_{n-1}u_{n-1}+v_{n-1}u_{n-2}(u_{n}+u_{n-2})-iD_{n}^{2}(\ln h_{n-1})_{t},
\notag \\
\sigma _{n}^{(1)} &=&i\frac{d}{dt}\sigma
_{n-1}^{(0)}+(h_{n-1}h_{n-2}-1)+a_{n-1}^{(1)}v_{n-1}(u_{n}+u_{n-2}),  \notag
\\
a_{n}^{(2)} &=&-3a_{n-1}^{(1)}+i\frac{d}{dt}\sigma
_{n-1}^{(1)}-ia_{n-1}^{(1)}D_{n}^{-1}(\ln h_{n-1})_{t}+  \notag \\
&&+a_{n}^{(1)}\sigma _{n}^{(0)}-u_{n-1}(v_{n}+v_{n-2})a_{n-1}^{(1)},  \notag
\\
dh_{n}/dt &=&iD_{n}(v_{n-1}u_{n}-v_{n}u_{n-1}),...,  \notag
\end{eqnarray}%
for all $n\in \mathbb{Z},s\in \mathbb{Z},$ \ or%
\begin{eqnarray}
\sigma _{n}^{(0)}
&=&v_{n-1}u_{n-1}+v_{n-1}u_{n-2}(u_{n}+u_{n-2})-iD_{n}^{2}(\ln h_{n-1})_{t},
\label{D69} \\
\sigma _{n}^{(1)} &=&i\frac{d}{dt}\sigma
_{n-1}^{(0)}+(1-v_{n-1}u_{n-1})(1-v_{n-2}u_{n-2})+\beta
v_{n-1}(u_{n}+u_{n-2}),...,  \notag
\end{eqnarray}%
and so on. Thus, having stated that the corresponding iterative equations
are solvable for all $s\in \mathbb{Z}_{+},$ we can claim that expression \ (%
\ref{D64}) is a true asymptotical solution to the Lax equation \ (\ref{D65a}%
).
\end{proof}

Recalling now that the expression
\begin{equation}
\gamma (\lambda ):=-\sum_{n=0}^{N-1}\ln h_{n}+\sum_{n=0}^{N-1}\ln
(1-\sum\limits_{s\in \mathbb{Z}_{+}}\sigma _{n}^{(s)}\lambda ^{-s-1})
\label{d70}
\end{equation}%
as $\lambda \rightarrow \infty $ is a generating function of conservation
laws for the dynamical system \ (\ref{D2}), one finds that functionals
\begin{eqnarray}
\bar{\gamma}_{0} &=&\sum_{n=0}^{N-1}\ln (1-v_{n}u_{n}),\gamma
_{0}=-\sum_{n=0}^{N-1}\sigma _{n}^{(0)},  \label{D70} \\
\gamma _{1} &=&-\sum_{n=0}^{N-1}(\sigma _{n}^{(1)}+\frac{1}{2}\sigma
_{n}^{(0)}\sigma _{n}^{(0)}),  \notag \\
\gamma _{2} &=&-\sum_{n=0}^{N-1}(\sigma _{n}^{(2)}+\frac{1}{3}\sigma
_{n}^{(0)}\sigma _{n}^{(0)}\sigma _{n}^{(0)}+\sigma _{n}^{(0)}\sigma
_{n}^{(1)}),...,  \notag
\end{eqnarray}%
and so on, make up an infinite hierarchy of exact conservative quantities
for the discrete nonlinear Schr\"{o}dinger dynamical system \ (\ref{D2}).

Make here some remarks concerning the complete integrability of the discrete
nonlinear Schr\"{o}dinger dynamical system \ (\ref{D2}). First we can easily
enough state, making use of the standard asymptotical small parameter
approach \cite{PM,BPS,HPP}, that the N\"{o}ther equation \ (\ref{D5}) on the
manifold $M_{(N)}$ possesses \cite{Pr,MS} the following exact Poissonian
operator solution:

\begin{equation}
\vartheta _{n}=\left(
\begin{array}{cc}
0 & ih_{n} \\
-ih_{n} & 0%
\end{array}%
\right) ,  \label{D71}
\end{equation}%
$n\in \mathbb{Z}_{N},$ subject to which the dynamical the dynamical system \
(\ref{D2}) is Hamiltonian:%
\begin{equation}
d(u,v)^{\intercal }/dt=-\vartheta \text{ }\mathrm{grad}\ H_{\vartheta }[u,v]
\label{D72}
\end{equation}%
on the periodic manifold $M_{(N)},$ where the Hamiltonian function
\begin{equation}
H_{\vartheta }:=\sum_{n=0}^{N-1}\ln
h_{n}^{2}-\sum_{n=0}^{N-1}(v_{n}u_{n-1}+v_{n-1}u_{n}-)=2\ln \bar{\gamma}_{0}-%
\frac{1}{2}(\gamma _{0}+\gamma _{0}^{\ast }).  \label{D73}
\end{equation}%
By means of similar, but more cumbersome calculations, one can find the
second Poissonian operator solution to the N\"{o}ther equation \ (\ref{D5})
in the following matrix form:%
\begin{eqnarray}
\eta _{n} &=&\left(
\begin{array}{cc}
(h_{n}-u_{n}D_{n}^{-1}u_{n})\Delta & (u_{n}^{2}+u_{n}D_{n}^{-1}u_{n})\Delta
^{-1} \\
v_{n}D_{n}^{-1}v_{n}\Delta & -(1+v_{n}D_{n}^{-1}u_{n})\Delta ^{-1}%
\end{array}%
\right) \times  \notag \\
&&\times \left(
\begin{array}{cc}
u_{n}D_{n}^{-1}u_{n} & (h_{n}-u_{n}D_{n}^{-1}v_{n} \\
1+v_{n}D_{n}^{-1}u_{n} & -(v_{n}+v_{n}D_{n}^{-1}v_{n})%
\end{array}%
\right) ,  \label{D74}
\end{eqnarray}%
where the operation $D_{n}^{-1}(...):=\frac{1}{2}[\sum%
\limits_{k=0}^{n-1}(...)_{k}-\sum\limits_{k=n}^{N-1}(...)_{k}]$ is
quasi-skew-symmetric with respect to the usual bi-linear form on $T^{\ast
}(M_{(N)})\times T(M_{(N)}),$ satisfying the operator identity $%
(D_{n}^{-1})^{\ast }=-\Delta ^{-1}D_{n}^{-1}\Delta ,$ $n\in \mathbb{Z}.$

The Poissonian operators \ (\ref{D71}) and \ (\ref{D74}) are compatible,
that makes it possible to construct by means of the algebraic
gradient-holonomic algorithm the related Lax type representation for the
dynamical system \ (\ref{D2}). The corresponding result is as follows: the
discrete linear spectral problem
\begin{equation}
\Delta f_{n}=l_{n}[u,v;\lambda ]f_{n},  \label{D75}
\end{equation}%
where $f\in l_{\infty }(\mathbb{Z};\mathbb{C}^{2})$ and for $n\in \mathbb{Z}$
\begin{equation}
l_{n}[u,v;\lambda ]=\left(
\begin{array}{cc}
\lambda & u_{n} \\
v_{n} & \lambda ^{-1}%
\end{array}%
\right) ,  \label{D76}
\end{equation}%
allows the linear Lax type isospectral evolution%
\begin{equation}
df_{n}/dt=p_{n}(l)f_{n}  \label{D77}
\end{equation}%
for some matrix $p_{n}(l)\in End$ $\mathbb{C}^{2},n\in \mathbb{Z},$
equivalent to the Hamiltonian flow
\begin{equation}
df_{n}/dt=\{H_{\vartheta },f_{n}\}_{\vartheta },  \label{D78}
\end{equation}%
where $\{.,.\}_{\vartheta }$ is the corresponding to \ (\ref{D71})
Poissonian structure on the manifold $M_{(N)}.$ \ The equivalence of \ (\ref%
{D71}) and \ (\ref{D78}) can be easily enough demonstrated, if to construct
the corresponding to \ (\ref{D75}) monodromy matrix $S_{n}(\lambda ),n\in
\mathbb{Z},$ for all $\lambda \in \mathbb{C}$ and to calculate the
Hamiltonian evolution
\begin{equation}
\frac{d}{dt}S_{n}(\lambda )=\{H_{\vartheta },S_{n}(\lambda )\}_{\vartheta
}=[p_{n}(l),S_{n}(\lambda )],  \label{D79}
\end{equation}%
giving rise to the same matrix $p_{n}(l)\in End\mathbb{C}^{2},n\in \mathbb{Z}%
,$ as that entering equation \ (\ref{D77}).

Thus, we have shown that the nonlinear discrete Schr\"{o}dinger dynamical
system \ (\ref{D2}) is Lax type integrable bi-Hamiltonian flow on the
manifold $M_{(N)}.$ Since the solution $\varphi (\lambda )\in T^{\ast
}(M_{(N)}),$ constructed above, satisfies the gradient-like relationship
\begin{equation}
\lambda \vartheta \text{ }\varphi (\lambda )=\eta \text{ }\varphi (\lambda )
\label{D80}
\end{equation}%
for all for $\lambda \in \mathbb{C},$ we derive that the found above
conservation laws are commuting to each other with respect to both Poisson
brackets $\{.,.\}_{\vartheta }$ and $\{.,.\}_{\vartheta }.$ The latter gives
rise to the classical Liouville integrability \cite{Ar,MBPS} of the discrete
nonlinear Schr\"{o}dinger dynamical system \ (\ref{D2}) on the periodic
manifold $M_{(N)}.$ The detail analysis of the integrability procedure via
the mentioned before Bogoyavlensky- Novikov reduction \cite{BN,No} and the
explicit construction of solutions to the dynamical system \ (\ref{D2}) are
planned to be presented in a separate work.

\section{Example 2: the Ragnisco-Tu differential-difference dynamical system
and its integrability}

Consider the Ragnisco-Tu differential-difference dynamical system\ (\ref{D2a}%
), defined on the periodic manifold $M_{(N)}\subset l_{\infty }(\mathbb{Z};%
\mathbb{C}^{2}),$ and construct first the corresponding asymptotical
solution to the Lax equation (\ref{D9}). The following lemma holds.

\begin{lemma}
\label{Lm_9} The functional expression%
\begin{equation}
\varphi _{n}:=\binom{a_{n}(\lambda )}{1}\exp (\lambda
t)\prod\limits_{j=1}^{n}\sigma _{j}(\lambda ),  \label{d81}
\end{equation}%
is an asymptotical, as $\ \ \lambda \rightarrow \infty ,$ \ solution to the
determining Lax equation (\ref{D9}) for all $n\in \mathbb{Z}_{N}$ with the
operator $K^{\prime ,\ast }:T^{\ast }(M_{(N)})\rightarrow T^{\ast }(M_{(N)})$
of the form:
\begin{equation}
K_{n}^{\prime ,\ast }=\left(
\begin{array}{cc}
\begin{array}{c}
\Delta ^{-1}-2u_{n}v_{n} \\
-u_{n}^{2}%
\end{array}
&
\begin{array}{c}
v_{n}^{2} \\
-\Delta +2u_{n}v_{n}%
\end{array}%
\end{array}%
\right) .  \label{d82}
\end{equation}%
where, by definition,%
\begin{eqnarray}
\sigma _{n}(\lambda ) &\simeq &\lambda (1-\sum\limits_{s\in \mathbb{Z}%
_{+}}\sigma _{n}^{(s)}[u,v]\lambda ^{-s}),  \label{d83} \\
a_{n}(\lambda ) &\simeq &\sum\limits_{s\in \mathbb{Z}_{+}}a_{n}^{(s)}[u,v]%
\lambda ^{-s},  \notag
\end{eqnarray}%
and the following analytical expressions%
\begin{eqnarray}
\sigma _{n}^{(0)} &=&0,a_{n}^{(0)}=0;\sigma
_{n}^{(1)}=-2u_{n-1}v_{n-1},a_{n}^{(1)}=-v_{n}^{2};  \notag \\
\sigma _{n}^{(2)}
&=&2u_{n-1}v_{n-2}-u_{n-1}^{2}v_{n-1}^{2},a_{n}^{(2)}=2v_{n}(v_{-n-1}-v_{n}^{2}u_{n});
\notag \\
\sigma _{n}^{(3)} &=&-2u_{n-1}v_{n-2}-D_{n}^{-1}(d\sigma
_{n}^{(2)}/dt+\sigma _{n}^{(1)}d\sigma _{n}^{(1)}/dt),  \label{d84} \\
a_{n}^{(3)}
&=&-da_{n}^{(2)}/dt-2(u_{n-1}v_{n-2}v_{n}^{2}-u_{n}v_{n}v_{n-1}^{2}),...,
\notag
\end{eqnarray}%
and so on, hold.
\end{lemma}

\begin{proof}
It is easy to calculate that local $\sigma $- and $a$-functionals on $%
M_{(N)} $ satisfy the following functional equations:%
\begin{equation}
\begin{array}{c}
\lambda (1-\sigma _{n}(\lambda ))+D_{n}^{-1}\frac{d}{dt}\ln \sigma
_{n}(\lambda )-u_{n}^{2}a_{n}(\lambda )+2u_{n}v_{n}=0, \\
da_{n}(\lambda )/dt+\lambda a_{n}(\lambda )+a_{n}(\lambda )D_{n}^{-1}\frac{d%
}{dt}\ln \sigma _{n}(\lambda )-2u_{n}v_{n}\lambda ^{-1}a_{n-1}(\lambda
)\sigma _{n}(\lambda )^{-1}+v_{n}^{2}=0,%
\end{array}
\label{d85}
\end{equation}%
which allow the asymptotical as $\lambda \rightarrow \infty $ solutions in
the form (\ref{d83}). Then, solving step-by-step the corresponding recurrent
equations one finds the exact analytical expressions (\ref{d84}). Taking now
into account that for each $n\in \mathbb{Z}_{+}$ there exists such a local
functional $\rho _{n}(\lambda )$ that the expression $\frac{d}{dt}\ln \sigma
_{n}(\lambda )=D_{n}\rho _{n}(\lambda )$ holds on $M_{(N)}$ identically, we
derive that the functional expression (\ref{d81}) solves the Lax equation (%
\ref{D9}), proving the lemma.
\end{proof}

As a simple corollary of Lemma \ref{Lm_9} we obtain that the expression
\begin{equation}
\gamma (\lambda ):=\sum_{n=1}^{N}\ln (1-\sum\limits_{s\in \mathbb{Z}%
_{+}}\sigma _{n}^{(s)}\lambda ^{-s-1})\simeq \sum_{j\in \mathbb{Z}%
_{+}}\gamma _{j}\lambda ^{-j}  \label{d87}
\end{equation}%
is a generating functional for the infinite hierarchy of conservation laws $%
\gamma _{j}\in D(M_{(N)}),j\in \mathbb{Z}_{+},$ of the Ragnisco-Tu
differential-difference dynamical system\ (\ref{D2a}).

Show now that the Ragnisco-Tu differential-difference dynamical system\ (\ref%
{D2a}) is a bi-Hamiltonian dynamical system on the functional manifold $%
M_{(N)}.$ Really, based on Lemma \ref{Lm_4}, we can find that the element $\
\psi :=\frac{1}{2}(v_{n},-u_{n})^{\intercal }\in T^{\ast }(M_{(N)})$
satisfies \ the functional equation (\ref{D14}):%
\begin{equation}
d\psi /dt+K^{^{\prime ,\ast }}\psi =\mathrm{grad}\text{ }\mathcal{L},\text{ }%
\mathcal{L}=\mathcal{-}\frac{1}{2}\sum_{k=0}^{N-1}u_{n}^{2}v_{n}^{2},
\label{d88}
\end{equation}%
giving rise to the first Poissonian structure
\begin{equation}
\vartheta _{n}:=\psi _{n}^{\prime }-\psi _{n}^{\prime ,\ast }=\left(
\begin{array}{cc}
0 & 1 \\
-1 & 0%
\end{array}%
\right)  \label{d89}
\end{equation}%
on the manifold $M_{(N)}\ $\ with respect to which the
differential-difference dynamical system\ (\ref{D2a}) is Hamiltonian:%
\begin{equation}
\frac{d}{dt}(u_{n},v_{n})^{\intercal }=-\vartheta _{n}\text{ }\mathrm{grad}%
\text{ }H_{\vartheta ,n}[u,v],  \label{d90}
\end{equation}%
where the Hamiltonian function, owing to relationship (\ref{D19}), equals
\begin{equation}
H_{\vartheta }:=(\psi ,K)-\mathcal{L}_{\psi
})=\sum_{k=0}^{N-1}(u_{n}^{2}v_{n}^{2}/2-u_{n}v_{n-1})=-\frac{1}{2}%
\sum_{k=0}^{N-1}\sigma _{n}^{(2)}.  \label{d91}
\end{equation}%
The same way one can find the second compatible with (\ref{d89}) Poissonian
operator
\begin{equation}
\eta _{n}:=\left(
\begin{array}{cc}
-u_{n}^{2}+2u_{n}D_{n}^{-1}\Delta u_{n} & \Delta -2u_{n}D_{n}^{-1}\Delta
v_{n} \\
-\Delta ^{-1}+2u_{n}v_{n}-2v_{n}D_{n}^{-1}\Delta u_{n} &
-v_{n}^{2}+2v_{n}D_{n}^{-1}\Delta v_{n}%
\end{array}%
\right) ,  \label{d92}
\end{equation}%
for which
\begin{equation}
\frac{d}{dt}(u_{n},v_{n})^{\intercal }=-\eta _{n}\text{ }\mathrm{grad}\text{
}H_{\eta ,n}[u,v],  \label{d93}
\end{equation}%
where the Hamiltonian function
\begin{equation}
H_{\eta }:=-\sum_{k=1}^{N}u_{n}v_{n}=\frac{1}{2}\sum_{k=1}^{N}\sigma
_{n+1}^{(1)}.  \label{d94}
\end{equation}%
Moreover, we claim that the hierarchy of conservation laws (\ref{d87})
satisfies as $\lambda \rightarrow \infty $ the gradient relationship
\begin{equation}
\lambda \vartheta \text{ }\mathrm{grad}\text{ }\gamma (\lambda )=\eta \text{
}\mathrm{grad}\text{ }\gamma (\lambda ),  \label{d95}
\end{equation}%
entailing their commutation with respect to both Poissonian structures (\ref%
{d89}) and (\ref{d92}). The latter allows us to argue that the Ragnisco-Tu
differential-difference dynamical system\ (\ref{D2a}) is a completely
integrable bi-Hamiltonian dynamical system on the manifold $M_{(N)}.$

Since the gradient relationship (\ref{d95}) \ gives rise to the following
`adjoint' Lax type representation
\begin{equation}
d\Lambda /dt=[\Lambda ,K^{^{\prime ,\ast }}],  \label{d96}
\end{equation}%
where, by definition, the expression $\Lambda :=\vartheta ^{-1}\eta :T^{\ast
}(M_{(N)})\rightarrow T^{\ast }(M_{(N)})$ is called a \textit{recursion }%
operator. \ Based on the gradient relationship (\ref{d95}) and expression (%
\ref{D29a}) we can obtain, within the gradient holonomic approach, that the
Ragnisco-Tu differential-difference dynamical system\ (\ref{D2a}) is also
Lax type integrable whose standard linear shift Lax type spectral problem
equals%
\begin{equation}
\Delta f_{n}=l_{n}[u,v;\lambda ]f_{n},\text{ \ \ \ }l_{n}[u,v;\lambda
]=\left(
\begin{array}{cc}
\lambda +u_{n}v_{n} & u_{n} \\
v_{n} & 1%
\end{array}%
\right) ,  \label{d97}
\end{equation}%
for all $n\in \mathbb{Z},\ \lambda \in $\ $\mathbb{C},$ \ where ($u,v)\in
M_{(N)}$ $\ $\ and $\ \ f\in l_{\infty }(\mathbb{Z};\mathbb{C}^{2}).$

\section{Example 3: the generalized differential-difference Riemann-Burgers
and Riemann type dynamical systems and their integrability}

In \cite{PAPP,GBPPP,GPPP} there was analyzed by means of simple enough
differential-algebraic tools a generalized (owing to D. Holm and M. Pavlov)
Riemann type hydrodynamical hierarchy of equations
\begin{equation}
D_{t}^{s}u=0,\text{ \ \ }D_{t}:=\partial /\partial t+uD_{x},D_{x}:=\partial
/\partial x,  \label{d98}
\end{equation}%
on a smooth functional manifold $\mathcal{\bar{M}}\subset C^{\infty }(%
\mathbb{R};\mathbb{R})$ for any integer $s\in \mathbb{Z}_{+}$ and proved
their both bi-Hamiltonian structure and Lax type integrability. For \ $s=2$
\ equation \ (\ref{d98}) possesses a Lax type representation, whose Lax $l$%
-operator is given \cite{GPPP} by the expression
\begin{equation}
l[u,v;\lambda ]=\left(
\begin{array}{cc}
-u_{x}\lambda & -v_{x} \\
2\lambda ^{2} & u_{x}\lambda%
\end{array}%
\right) ,  \label{d99}
\end{equation}%
where the vector-function $(u,v)^{\intercal }\in \mathcal{M}\subset
C^{\infty }(\mathbb{R};\mathbb{R}^{2})$ satisfies the equivalent to \ (\ref%
{d98}) nonlinear dynamical system
\begin{equation}
D_{t}u=v,\text{ \ \ }D_{t}v=0,  \label{d100}
\end{equation}%
and $\lambda \in \mathbb{C}$ is an arbitrary invariant spectral parameter.
For studying differential-difference versions of equations \ (\ref{d100}) we
will apply to the Lax $l$-operator \ (\ref{d99}) the Ablowitz-Ladik
discretization scheme \cite{AL}. As a result of simple calculations we
obtain the following new discrete Lax type spectral problem:
\begin{equation}
f_{n+1}=l_{n}[u,v;\lambda ]f_{n},\text{ \ \ \ \ \ }l_{n}[u,v;\lambda
]=\left(
\begin{array}{cc}
1-\lambda D_{n}u_{n} & -D_{n}v_{n} \\
2\lambda ^{2} & 1+\lambda D_{n}u_{n}%
\end{array}%
\right)  \label{d101}
\end{equation}%
for $n\in \mathbb{Z},$\ where function $f\in l_{\infty }(\mathbb{Z};\mathbb{C%
}^{2}),$ the vector $(u,v)^{\intercal }\in M,$ if the resulting discrete
dynamical system is considered on an $N$-periodical discrete manifold $%
M\subset (\mathbb{R}^{2})^{\mathbb{Z}_{N}}$\ $.$

To study the related with \ the spectral problem (\ref{d101}) nonlinear
differential-difference dynamical systems, we will make use of a slightly
generalized gradient-holonomic scheme.

First, we formulate the following simple but useful lemma.

\begin{lemma}
\ \label{Lm_10} The following matrix
\begin{equation}
\tilde{F}_{m,n}(\lambda )=\left(
\begin{array}{cc}
\tilde{e}_{m,n}^{(1)}(\lambda ) & -D_{m}\hat{u}_{m}(\lambda )\tilde{e}%
_{m,n}^{(2)}(\lambda )/(2\lambda ) \\
-2\lambda \tilde{e}_{m,n}^{(1)}(\lambda )/D_{m}\check{u}_{m}(\lambda ) &
\tilde{e}_{m,n}^{(2)}(\lambda )%
\end{array}%
\right) ,  \label{d102}
\end{equation}%
where%
\begin{eqnarray}
\tilde{e}_{m,n}^{(1)}(\lambda ) &:&=\prod\limits_{k=n}^{m-1}[1-\lambda
D_{k}u_{k}+\lambda (D_{k}u_{k})^{2}/D_{k}\check{u}_{k}(\lambda )],  \notag \\
\tilde{e}_{n,m}^{(2)}(\lambda ) &:&=\prod\limits_{k=n}^{m-1}[1+\lambda
(D_{k}u_{k}-D_{k}\hat{u}_{k}(\lambda ))],  \label{d103}
\end{eqnarray}%
jointly with functional relationships%
\begin{equation}
\binom{D_{n}\check{u}_{n}(\lambda )}{D_{n}\hat{u}_{n}(\lambda )}=\binom{%
D_{n}u_{n}+\lambda ^{-1}[1+\frac{D_{n}\check{u}_{n}(\lambda )-D_{n}u_{n}}{%
D_{n}^{2}\check{u}_{n}(\lambda )}]^{-1}}{D_{n}u_{n}+\lambda ^{-1}[1+\frac{%
D_{n}\hat{u}_{n}(\lambda )-D_{n}u_{n}}{D_{n}^{2}\hat{u}_{n}(\lambda )}]^{-1}}
\label{d104}
\end{equation}%
and
\begin{equation}
2D_{n}v_{n}-(D_{n}u_{n})^{2}=0,\   \label{d104a}
\end{equation}%
solves the associated with spectral problem \ (\ref{d101}) \ linear matrix
equation
\begin{equation}
\tilde{F}_{m+1,n}(\lambda )=l_{m}[u,v;\lambda ]\tilde{F}_{m,n}(\lambda )
\label{d105}
\end{equation}%
under the initial condition%
\begin{equation}
\tilde{F}_{m,n}(\lambda )|_{m=n}=\mathbf{I+}O\mathbf{(}1/\lambda )
\label{d106}
\end{equation}%
as $\lambda \rightarrow \infty $ for all $m,n\in \mathbb{Z}.$
\end{lemma}

As a corollary from relationships \ (\ref{d104}) one easily finds limits
\begin{equation}
\lim_{\lambda \rightarrow \infty }D_{n}\check{u}_{n}(\lambda )=D_{n}u_{n},%
\text{ \ \ }\lim_{\lambda \rightarrow \infty }D_{n}\hat{u}_{n}(\lambda
)=D_{n}u_{n},  \label{d107}
\end{equation}%
uniformly holding for all $n\in \mathbb{Z}.$

Now, based on the matrix expressions (\ref{d102}) and (\ref{d103}), one can
construct the fundamental matrix%
\begin{equation}
F_{m,n}(\lambda ):=\tilde{F}_{m,n}(\lambda )\tilde{F}_{n,n}^{-1}(\lambda ),
\label{d108}
\end{equation}%
solving the linear problem
\begin{equation}
F_{m+1,n}(\lambda )=l_{m}[u,v;\lambda ]F_{m,n}(\lambda )  \label{d109}
\end{equation}
under the initial condition
\begin{equation}
F_{m,n}(\lambda )|_{m=n}=\mathbf{I}  \label{d110}
\end{equation}%
for all $n\in \mathbb{Z}$ as $\lambda \rightarrow \infty .$

\bigskip Taking into account that the manifold $M$ is $N$ - periodic, one
can construct the next important objects - the asymptotical as $\lambda
\rightarrow \infty $ \ monodromy matrix%
\begin{equation}
S_{n}(\lambda ):=F_{n+N,n}(\lambda )  \label{d111}
\end{equation}%
for any $n\in \mathbb{Z}.$ By construction, the monodromy matrix (\ref{d111}%
) satisfies the following useful properties:%
\begin{equation}
S_{n+N}(\lambda ):=S_{n}(\lambda ),\text{ \ \ }\det S_{n}(\lambda )=1,
\label{d112}
\end{equation}%
holding for all $n\in \mathbb{Z}$ and $\lambda \rightarrow \infty .$

Keeping now in mind the importance of invariants and Poissonian structures
related with the linear spectral problem (\ref{d101}), we proceed to
studying its basic Lie-algebraic properties and connections with the so
called vertex operator representation \cite{JM,Ne} of the related whole
hierarchy of integrable differential-difference dynamical systems on the
manifold $M.$

Namely, we will sketch below the Lie-algebraic aspects \cite%
{FT,Ne,RS-T,RS-T1} of differential-difference dynamical systems, associated
with our Lax-type linear difference spectral problem \ (\ref{d101}), where \
one assumes that the matrix $\ $\ $l_{n}:=l_{n}[u,v;\lambda ]\in G_{n}:=$ $%
GL_{2}(\mathbb{C)\otimes C}(\lambda ,\lambda ^{-1})$ for $n\in \mathbb{Z}_{N}
$ $:=\mathbb{Z}/N\mathbb{Z}$ \ as $\lambda \rightarrow \infty .$ To describe
the related Lax type integrable dynamical systems, \ define first the matrix
product-group $G^{N}:=\overset{N}{\underset{j=1}{\otimes }}G_{j}$ and its
action $G^{N}\times M_{G}^{(N)}\rightarrow M_{G}^{(N)}$ \ \ on the phase
space $M_{G}^{(N)}:=\{l_{n}\in G_{n}:n\in $ \bigskip $\mathbb{Z}_{N}\},$ \
given as%
\begin{equation}
\{g_{n}\in G_{n}:n\in \mathbb{Z}_{N}\}\times \{l_{n}\in G_{n}:n\in \mathbb{Z}%
_{N}\}=\{g_{n}l_{n}g_{n+1}^{-1}\in G_{n}:n\in \mathbb{Z}_{N}\}.  \label{d113}
\end{equation}%
Subject to action ( \ref{d113}) a functional $\gamma \in \mathcal{D}%
(M_{G}^{(N)})$ is invariant iff the following discrete relationship%
\begin{equation}
\mathrm{grad}\gamma (l_{n})l_{n}=l_{n+1}\mathrm{grad}\gamma (l_{n+1})
\label{d114}
\end{equation}%
holds for all $n\in \mathbb{Z}_{N}.$

Assume further that the matrix group $G^{N}$ is identified with its tangent
spaces $T_{l}(G^{N}),$ $l\in G^{N},$ locally isomorphic to the Lie  algebra $%
\mathcal{G}^{(N)},$ where $\mathcal{G}^{(N)}$ is the corresponding Lie
algebra of the Lie group  $G^{N},$ isomorphic, by definition, to the tangent
space $T_{e}(G^{N})$ at the group unity $e\in G^{N}.$ With  any element $%
l\in G^{N}$ there are associated,  respectively,  the left $\eta _{l}:%
\mathcal{G}^{(N)}\rightarrow T_{l}(G^{N})$ and right $\ \ \rho _{l}:\mathcal{%
G}^{(N)}\rightarrow T_{l}(G^{N})$ differentials of \ the left and right
shifts  on the Lie group $G^{N},$ and their adjoint mappings $\rho
_{l}^{\ast }:T_{l}^{\ast }(G^{N})\rightarrow \mathcal{G}^{(N),\ast }$ and $%
\eta _{l}^{\ast }:T_{l}^{\ast }(G^{N})\rightarrow \mathcal{G}^{(N),\ast },$
\ where%
\begin{eqnarray}
(\rho _{l}^{\ast }\mathrm{grad}\gamma (l),X) &=&(\mathrm{grad}\gamma
(l),Xl)=(l\text{ }\mathrm{grad}\gamma (l),X):=tr(l\text{ }\mathrm{grad}%
\gamma (l)X),\text{ \ \ }  \notag \\
(\eta _{l}^{\ast }\mathrm{grad}\gamma (l),X) &=&(\mathrm{grad}\gamma
(l),lX)=(\mathrm{grad}\gamma (l)l,X):=tr(\mathrm{grad}\gamma (l)lX)
\label{d115}
\end{eqnarray}%
for any   $X\in \mathcal{G}^{(N)}$ and smooth functional $\gamma \in
\mathcal{D}(G^{N}),$ $tr:G^{N}\rightarrow \mathbb{C}$ is a   trace-operation
on the group $G^{N}:trA:=res_{\lambda =\infty }\sum\limits_{j\in \mathbb{Z}%
_{N}}SpA_{j}[u,v;\lambda ]$ for any $A\in G^{N}.$ Owing to (\ref{d114}) and (%
\ref{d115}) we can define the set%
\begin{equation}
\{\Phi _{n}=\mathrm{grad}\gamma (l_{n})l_{n}\mathrm{\ }\in \mathcal{G}%
_{n}^{\ast }:=T_{e}^{\ast }(G),\text{ \ }n\in \mathbb{Z}_{N}\}  \label{d116}
\end{equation}%
belonging to the space  $\mathcal{G}^{(N),\ast }\simeq T_{e}^{\ast }(G^{N})$
and satisfying the following invariance property:%
\begin{equation}
\Phi _{n+1}=Ad_{l_{n}}^{\ast }\Phi _{n}(\lambda )=l_{n}^{-1}\Phi
_{n}(\lambda )l_{n}  \label{d117}
\end{equation}%
for any $n\in \mathbb{Z}_{N}.$ The relationship (\ref{d117}) allows to
define a function $\varphi :G^{N}\rightarrow \mathbb{C}$ invariant with
respect to the adjoint action
\begin{equation}
G_{n}\times G_{n}\ni (g,S_{n}(\lambda ))\rightarrow ad_{g}S_{n}(\lambda
)=gS_{n}(\lambda )g^{-1}\in G_{n}  \label{d118}
\end{equation}%
for any $n\in \mathbb{Z}_{N}$ and such that
\begin{equation}
\gamma (l)=\varphi \lbrack S_{N}(\lambda )],\text{ }\Phi _{N}=\mathrm{grad}%
\varphi \lbrack S_{N}(\lambda )]S_{N}(\lambda ),  \label{d119}
\end{equation}%
where, by definition, the expression%
\begin{equation}
S_{N}(\lambda )=\overset{\overset{N}{\longleftarrow }}{\underset{j=1}{\Pi }}%
l_{j}[u,v;\lambda ]  \label{d120}
\end{equation}%
coincides exactly with the proper monodromy matrix for the linear spectral
problem (\ref{d101}). Since, owing to (\ref{d117}), the matrices $\Phi _{n}=%
\mathrm{grad}\varphi \lbrack S_{n}(\lambda )]S_{n}(\lambda )\in \mathcal{G}%
_{n}^{\ast },$ $n\in \mathbb{Z}_{N},$ can be reconstructed from (\ref{d120}%
), we find  \cite{FT,RS-T1} the following Poissonian flow on the matrices $%
S_{n}(\lambda )\in G_{n},n\in \mathbb{Z}_{N}:$%
\begin{equation}
dS_{n}(\lambda )/dt=[\mathcal{R}(\mathrm{grad}\varphi \lbrack S_{n}(\lambda
)]S_{n}(\lambda )),S_{n}(\lambda )]  \label{d121}
\end{equation}
with respect to  the invariant Casimir function $\varphi \in I(\mathcal{G}%
_{n}^{\ast })$ and the quadratic Poissonian structure
\begin{equation}
\{\gamma _{1},\gamma _{2}\}:=(l,[\mathrm{grad}\gamma _{1}(l),\mathcal{R}(l%
\text{ }\mathrm{grad}\gamma _{2}(l))]+[\mathcal{R}(l\text{ }\mathrm{grad}%
\gamma _{1}(l)),\mathrm{grad}\gamma _{2}(l)])  \label{d121a}
\end{equation}%
for any functionals $\gamma _{1},\gamma _{2}\in \mathcal{D}(G^{N}),$
constructed by means of a skew-symmetric $\mathcal{R}$-structure $\mathcal{R}%
:\mathcal{G}^{(N),\ast }\rightarrow \mathcal{G}^{(N)}.\ $ In particular,
the equality
\begin{equation}
\lbrack \mathrm{grad}\varphi (S_{n}),S_{n}]=0  \label{d122}
\end{equation}%
holds for all $n\in \mathbb{Z}_{N}.$

Having taken into account (\ref{d119}), one can rewrite (\ref{d121}) in the
following equivalent form:%
\begin{equation}
dS_{n}/dt=[\mathcal{R}(\mathrm{grad}\gamma (l_{n})l_{n}),S_{n}],
\label{d124}
\end{equation}%
holding for all $n\in \mathbb{Z}_{N}.$ The latter jointly with (\ref{d117})
makes it possible to retrieve \cite{JM,RS-T} the related evolution of
elements $l_{n}\in G_{n},$ $n\in \mathbb{Z}_{N}:$%
\begin{eqnarray}
dl_{n}/dt &=&p_{n+1}(l)l_{n}-l_{n}p_{n}(l),  \label{d125} \\
p_{n}(l) &:&=\mathcal{R}(\mathrm{grad}\gamma (l_{n})l_{n})  \notag
\end{eqnarray}%
and following from the relationships%
\begin{eqnarray}
S_{n}(\lambda ) &=&\psi _{n}(l)S_{N}(\lambda )\psi _{n}^{-1}(l),
\label{d126} \\
\psi _{n}(l) &=&\overset{\overset{n}{\longleftarrow }}{\underset{j=1}{\Pi }}%
l_{j}[u,v;\lambda ].  \notag
\end{eqnarray}%
Subject to the linear spectral problem (\ref{d101}) the solution $f\in
l_{\infty }(\mathbb{Z},\mathbb{C}^{2})$ satisfies the associated temporal
evolution equation%
\begin{equation}
df_{n}/dt=p_{n}(l)f_{n}  \label{d127}
\end{equation}%
for any $n\in \mathbb{Z}.$ It is easy to check that the compatibility
condition of linear equations (\ref{d101}) and (\ref{d127}) is equivalent to
the discrete Lax type representation (\ref{d125}), which, upon reducing it
on the group manifold $M_{G},$ gives rise to the corresponding nonlinear Lax
type integrable dynamical system on the discrete manifold $M.$ It follows
from the fact that all Casimir invariant functions, when reduced on the
manifold $M_{G},$ are in involution \cite{RS-T,RS-T1} \ with respect to the
Poisson bracket \ (\ref{d121a}).

In the case when the skew-symmetric $\mathcal{R}$-structure $\mathcal{R}%
=1/2(P_{+}-P_{-}),\ $ where $P_{\pm }:\mathcal{G}_{n}\rightarrow \mathcal{G}%
_{n,\pm }\subset \mathcal{G}_{n}$ are the projectors on the $\lambda $%
-positive and $\lambda $- negative, respectively, degree subalgebras of the
Lie algebra $\mathcal{G}_{n},$ the determining Lax type equation (\ref{d125}%
) generates the flows%
\begin{equation}
\frac{d}{dt_{j}}l_{n}[u,v;\lambda ]=(\lambda ^{j+1}\tilde{S}_{n+1}(\lambda
))_{+}l_{n}[u,v;\lambda ]-l_{n}[u,v;\lambda ](\lambda ^{j+1}\tilde{S}%
_{n}(\lambda ))_{+}  \label{d129}
\end{equation}%
for all $j\in \mathbb{Z}_{+},$ where $\tilde{S}_{n}(\lambda ),$ $n\in
\mathbb{Z}_{N},$ are the corresponding asymptotical expansion of the
suitably normalized monodromy matrix $\tilde{S}_{n}(\lambda )\in \mathcal{G}%
_{-},n\in \mathbb{Z}_{N},$ as $\lambda \rightarrow \infty .$

\bigskip\ The hierarchies of evolution equations (\ref{d129}) can be
rewritten as the following generating flows:%
\begin{equation}
\frac{d}{dt_{(\mu )}}l_{n}[u,v;\lambda ]=\frac{\lambda \mu }{\mu -\lambda }[%
\tilde{S}_{n+1}(\mu )l_{n}[u,v;\lambda ]-l_{n}[u,v;\lambda ]\tilde{S}%
_{n}(\mu )]  \label{d132}
\end{equation}%
as $\lambda \rightarrow \infty $ and $\left\vert \lambda /\mu \right\vert
<1, $ \ \ where, by definition,
\begin{equation}
\frac{d}{dt_{(\mu )}}=\underset{j\in \mathbb{Z}_{+}}{\sum }\mu ^{-j}\frac{d}{%
dt_{j}}.  \label{d134}
\end{equation}%
Proceed now to describing the analytical structure of the regularized
matrices $\tilde{S}_{n}(\mu ),$ $n\in \mathbb{Z}_{N},$ for arbitrary $\mu
\in \mathbb{C}.$ To do this effectively we need to consider the
corresponding to (\ref{d132}) evolution equations for the monodromy matrix $%
S_{n}(\lambda )$ as $\lambda \rightarrow \infty ,$ which make it possible to
construct the related $\ $flows on functions $\check{a}_{n}:=D_{n}\check{u}%
_{n}(\lambda )$ and $\hat{a}_{n}:=D_{n}\hat{u}_{n}(\lambda ),n\in \mathbb{Z}%
_{N},$ \ represented \cite{JM,Ne} by means of the related vertex operators \
action $\hat{X}_{\lambda }:\bar{M}^{2}\rightarrow \bar{M}^{2},$ where $\ \
\bar{M}:=\mathbb{R}^{\mathbb{Z}_{N}},$
\begin{eqnarray}
\text{\ }\hat{X}_{\lambda } &:&=(\exp D_{\lambda },\exp (-D_{\lambda
}))^{\intercal },  \notag \\
D_{\lambda } &:&=\underset{j\in \mathbb{Z}_{+}}{\sum }\frac{1}{(j\lambda
^{j})}\frac{d}{dt_{j}},  \label{d135}
\end{eqnarray}%
and
\begin{equation}
\hat{X}_{\lambda }\left(
\begin{array}{c}
u \\
u%
\end{array}%
\right) =\left(
\begin{array}{c}
u(t_{0}+1/\lambda ,t_{1}+1/(2\lambda ^{2}),...) \\
u(t_{0}-1/\lambda ,t_{1}-1/(2\lambda ^{2}),...)%
\end{array}%
\right) ,  \label{d136}
\end{equation}%
as $\lambda \rightarrow \infty .$

Now, based on the flows \ref{d132}, one can derive the corresponding
evolution equation on the monodromy matrix $\tilde{S}_{n}(\lambda )\in
\mathcal{G}_{-},n\in \mathbb{Z}_{N},$ \ with respect to the vector field \ (%
\ref{d134}):%
\begin{equation}
\frac{d}{dt_{(\mu )}}\tilde{S}_{n}(\lambda )=\frac{\lambda \mu }{\mu
-\lambda }[(\tilde{S}_{n}(\mu ),\tilde{S}_{n}(\lambda )]  \label{d137}
\end{equation}%
as $\left\vert \lambda /\mu \right\vert <1,\lambda \rightarrow \infty ,$
which entails upon taking the limit $\mu \rightarrow \lambda $ the equation
\begin{equation}
\frac{d}{dt}\tilde{S}_{n}(\lambda )=\lambda ^{2}[\frac{d}{d\lambda }\tilde{S}%
_{n}(\lambda ),\tilde{S}_{n}(\lambda )],  \label{d138}
\end{equation}%
where we put, by definition,
\begin{equation}
\frac{d}{dt}:=\underset{j\in \mathbb{Z}_{+}}{\sum }\lambda ^{-j}\frac{d}{%
dt_{j}}.  \label{d139}
\end{equation}%
Now, taking into account Lemma \ \ref{Lm_10}, the related matrix \
expressions (\ref{d102}), \ (\ref{d108}) and \ (\ref{d111}) and having
analyzed the analytical structure of the resulting monodromy $\tilde{S}%
_{n}(\lambda )\in \mathcal{G}_{-},n\in \mathbb{Z}_{N},$ one can state by
means of simple but slightly cumbersome calculations the following
proposition.

\begin{proposition}
\label{Prop_11} The following differential relationships
\begin{equation}
\binom{\frac{d}{dt}\left[ \frac{2\lambda s_{n}^{(11)}}{s_{n}^{(21)}}\left( 1-%
\sqrt{1-\frac{s_{n}^{(12)}s_{n}^{(21)}}{s_{n}^{(11)}s_{n}^{(11)}}}\right) %
\right] =-\lambda ^{2}\frac{d}{d\lambda }\left[ \frac{2\lambda s_{n}^{(11)}}{%
s_{n}^{(21)}}\left( 1-\sqrt{1-\frac{s_{n}^{(12)}s_{n}^{(21)}}{%
s_{n}^{(11)}s_{n}^{(11)}}}\right) \right] ,}{\frac{d}{dt}\left[ -\frac{%
2\lambda s_{n}^{(11)}}{s_{n}^{(21)}}\left( 1+\sqrt{1-\frac{%
s_{n}^{(12)}s_{n}^{(21)}}{s_{n}^{(11)}s_{n}^{(11)}}}\right) \right] =\lambda
^{2}\frac{d}{d\lambda }\left[ -\frac{2\lambda s_{n}^{(11)}}{s_{n}^{(21)}}%
\left( 1+\sqrt{1-\frac{s_{n}^{(12)}s_{n}^{(21)}}{s_{n}^{(11)}s_{n}^{(11)}}}%
\right) \right] ,},  \label{d140}
\end{equation}%
jointly with \ equalities
\begin{equation}
\begin{array}{c}
D_{n}\hat{u}_{n}(\lambda )=\frac{2\lambda s_{n}^{(11)}}{s_{n}^{(21)}}\left(
1-\sqrt{1-\frac{s_{n}^{(12)}s_{n}^{(21)}}{s_{n}^{(11)}s_{n}^{(11)}}}\right) ,
\\
\\
D_{n}\check{u}_{n}\left( \lambda \right) =-\frac{2\lambda s_{n}^{(11)}}{%
s_{n}^{(21)}}\left( 1+\sqrt{1-\frac{s_{n}^{(12)}s_{n}^{(21)}}{%
s_{n}^{(11)}s_{n}^{(11)}}}\right) ,%
\end{array}
\label{d141}
\end{equation}%
hold as $\lambda \rightarrow \infty $ \ for all $\ n\in \mathbb{Z}_{N}.$
\end{proposition}

As a corollary of \ (\ref{d107}) and differential relationships \ (\ref{d140}%
) one obtains easily the following \ vertex operator representation \ (\ref%
{d136}) for the functions (\ref{d141}):%
\begin{equation}
\left(
\begin{array}{c}
D_{n}\hat{u}_{n}(\lambda ) \\
D_{n}\check{u}_{n}(\lambda )%
\end{array}%
\right) =\left(
\begin{array}{c}
D_{n}u_{n}(t_{0}+1/\lambda ,t_{1}+1/(2\lambda ^{2}),...) \\
D_{n}u_{n}(t_{0}-1/\lambda ,t_{1}-1/(2\lambda ^{2}),...)%
\end{array}%
\right) =\hat{X}_{\lambda }\left(
\begin{array}{c}
D_{n}u_{n} \\
D_{n}u_{n}%
\end{array}%
\right) ,  \label{d142}
\end{equation}%
holding as $\lambda \rightarrow \infty $ \ for all $\ n\in \mathbb{Z}_{N}.$
Recalling additionally the algebraic \ relationships (\ref{d104}) and \ (\ref%
{d104a}), one finds from \ (\ref{d142}) the \ related infinite hierarchy of
differential-difference dynamical systems on the manifold $M:$%
\begin{equation}
\frac{d}{dt_{j}}\left(
\begin{array}{c}
D_{n}u_{n} \\
D_{n}v_{n}%
\end{array}%
\right) =\left(
\begin{array}{c}
c_{j}(D_{n}^{2}u_{n})^{-j} \\
2c_{j}v_{n}(D_{n}u_{n})^{-1}(D_{n}^{2}u_{n})^{-j}%
\end{array}%
\right)  \label{d143}
\end{equation}%
for all $j\in \mathbb{Z}_{+},$ where $c_{j}\in \mathbb{R},j\in \mathbb{Z}%
_{+},$ \ are some recurrently calculated constant coefficients.

Thus, we see that the constructed hierarchy \ (\ref{d143}) does not contain
a "naive" \ discretization of the generalized Riemann type system of
equations \ (\ref{d100}) in spite of the fact that the difference Lax type
linear spectral problem \ (\ref{d101}) regularly reduces to the
corresponding linear spectral problem for dynamical system \ (\ref{d100}). \

Nonetheless, such a situation is not faced in the case of the inviscid
discrete Riemann-Burgers dynamical system \ (\ref{D2b}):
\begin{equation}
dw_{n}/dt=w_{n}(w_{n+1}-w_{n-1})/2:=K_{n}[w],  \label{d144}
\end{equation}%
defined on an $N$-periodical discrete manifold $M\subset l_{2}(%
\mathbb{Z}
_{N};\mathbb{R}).$ Following the gradient-holonomic scheme of studying the
integrability of \ (\ref{d144}), we first state the existence of an infinite
hierarchy of conservation laws and the corresponding bi-Hamiltonian
formulation.

Consider, owing to Proposition \label{Prop_5}, the determining equation \ (%
\ref{D9})
\begin{equation}
d\varphi _{_{n}}/dt-[(\Delta -\Delta
^{-1})w_{n}/2-(w_{n+1}-w_{n-1})/2]\varphi _{_{n}}=0  \label{d145}
\end{equation}%
and \ its asymptotical solution $\varphi \in T^{\ast }(M)$ in the form \ (%
\ref{D31}):
\begin{equation}
\varphi _{_{n}}=\prod\limits_{j=0}^{n-1}\sigma _{j}[w;\lambda ],
\label{d146}
\end{equation}%
where $n\in
\mathbb{Z}
$ and the local functionals $\sigma _{j}[w;\lambda ],j\in
\mathbb{Z}
_{+},$ \ \ possess as $\lambda \rightarrow \infty $ the \ following
expansions
\begin{equation}
\sigma _{j}[w;\lambda ]\simeq \sum_{s\in
\mathbb{Z}
_{+}}\sigma _{j}^{(s)}[w]\lambda ^{-s}.  \label{d147}
\end{equation}%
Having solved recurrently the resulting functional equations%
\begin{equation}
D_{n}^{-1}(\ln \sigma _{n})_{t}+(w_{n-1}/\sigma _{n-1}-w_{n+1}\sigma
_{n})/2+(w_{n+1}-w_{n-1})=0,  \label{d148}
\end{equation}%
one finds easily the infinite hierarchy \ (\ref{d87}) of conservations laws:
\begin{equation}
\gamma _{0}=\sum_{n=0}^{N-1}(w_{n}+w_{n-1}),\gamma _{1}=0,\gamma
_{2}=\sum_{n=0}^{N-1}[(w_{n}+w_{n-1})^{2}+w_{n}(w_{n-1}+w_{n+1})],...,\gamma
_{2j+1}=0  \label{d149}
\end{equation}%
for all $j\in
\mathbb{Z}
_{+}.$ Now, applying to the hierarchy of conservation laws the approach of
Lemma \ \ref{Lm_4}, one can find by means of \ slightly cumbersome and
lengthy calculations the following pair \ $\vartheta ,\eta :T^{\ast
}(M)\rightarrow T(M)$ of compatible Poissonian operators on the manifold $M:$%
\begin{equation}
\vartheta _{n}:=w_{n}(\Delta -\Delta ^{-1})w_{n},\text{ \ }\eta
_{n}:=(w_{n}w_{n+1}\Delta ^{2}-w_{n}w_{n-1}\Delta ^{-2})(w_{n}+w_{n-1}\Delta
^{-1}).  \label{d150}
\end{equation}%
In particular, there easily obtains the Hamiltonian representation of the
Burgers-Riemann system \ (\ref{d144}):%
\begin{equation}
dw_{n}/dt=-\vartheta _{n}\mathrm{grad}\ {H}_{\vartheta },\text{ }%
H_{\vartheta }:=-\sum_{n=0}^{N-1}(w_{n}+w_{n-1})/2.  \label{d151}
\end{equation}%
Moreover, the first Poissonian structure of \ (\ref{d150}) allows the
continuous limit $\underset{_{\substack{ \Delta x\rightarrow 0 \\ %
n\rightarrow \infty }}}{\lim }w_{n}:=w(x),$ if $n\Delta x:=x\in \mathbb{R},$
to the well known \cite{Kup1} \ correct form
\begin{equation}
\vartheta :=(w\partial +\partial w)(w+\partial ^{-1}w\partial )/2.
\label{d152}
\end{equation}%
Making use the Poissonian pair \ (\ref{d150}) one can retrieve within the
gradient holonomic scheme a Lax type representation related to the inviscid
discrete Riemann-Burgers dynamical system \ \ (\ref{d144}), whose $l$%
-operator is given by the following \ matrix expression:%
\begin{equation}
\text{\ }l_{n}[w;\lambda ]=\left(
\begin{array}{cc}
\lambda  & -w_{n} \\
1 & 0%
\end{array}%
\right)   \label{d153}
\end{equation}%
for $n\in
\mathbb{Z}
$ and $\lambda \in \mathbb{C}.$ Mention only that the higher flows,
generated by the inviscid Burgers-Riemann dynamical system \ (\ref{d144}), \
have nothing to do with the generalized Riemann type hydrodynamic systems \ (%
\ref{d98}) and their discrete approximations. Thereby, it is necessary to
develop a different approach to constructing their proper integrable
discrete Lax type representation, compatible with the related continuous
limit.

\section{Conclusion}

The gradient-holonomic scheme of direct studying Lax type integrability of
differential-difference nonlinear dynamical systems described in this work
appears to be effective enough for applications in the one-dimensional case,
similar to the case \cite{PM,MBPS,BPS,HPP,Pr,MS} of nonlinear dynamical
systems defined on spatially one-dimensional functional manifolds. This
algorithm makes it possible to \ construct simply enough an infinite
hierarchy of conservation laws as well as to calculate their compatible
co-symplectic structures. As it was \ also shown, the reduced via the
Bogoyavlensky-Novikov approach integrable Hamiltonian dynamical systems on
the corresponding invariant periodic submanifolds generate finite
dimensional Liouville integrable Hamiltonian systems with respect to the
canonical Gelfand-Dikiy type symplectic structures. As \ interesting
examples \ the \ complete integrability analysis of the nonlinear discrete
Schr\"{o}dinger, the Ragnisco-Tu and Burgers-Riemann type dynamical systems
was presented.

Subject to different not-direct approaches to studying the integrability of
differential-difference dynamical systems on discrete manifolds it is worth
to mention the works \cite{Kup,Kup1,DT,MV,BG,BSzP,HPS,Mi1,HS,Wi} based on
the inverse spectral transform and related Lie-algebraic methods \cite%
{LWY,PZ,HPS,Mi1,HS,Wi}, where there are constructed \textit{a priori }both
conservation laws and\textit{\ }the \ corresponding Lax type
representations. \ Concerning these approaches, many of their important
analytical properties were constructively absorbed \ by the
gradient-holonomic scheme of this work and realized directly as an algorithm.

\section{Acknowledgements}

The second author (N.B.) is very indebted to the Abdus Salam International
Centre for Theoretical Physics, Trieste, Italy, for the kind hospitality
during his ICTP-2010 research scholarship. A.P. is cordially indebted to
Profs. Denis Blackmore (NJIT, Newark, USA) and Maksim Pavlov (Lebedev
Physics Institute of RAS, Moscow, Russia) for valuable comments and remarks
concerning the gradient-holonomic scheme of integrability analysis and \ its
application to the generalized (owing to D. Holm and M. Pavlov) Riemann type
hydrodynamical hierarchy of dynamical systems. \ Authors are also cordially
appreciated to Mrs. Natalia Prykarpatska (Lviv, Ukraine) and Dilys Grilli
(Trieste, Publications Office, ICTP) for professional help in editing and
preparing the manuscript for publication.

\bigskip

\end{document}